\documentclass[aps,pra,onecolumn,superscriptaddress,nofootinbib,10pt]{revtex4-2}

\usepackage{amsmath,amsthm,amsfonts,amssymb,mathtools}
\usepackage{physics}
\usepackage{bbold}
\usepackage{dsfont}
\usepackage{graphicx}
\usepackage{booktabs}
\usepackage{xcolor}
\usepackage{hyperref}  
\hypersetup{
    colorlinks=false,      
    pdfborder={0 0 1},     
    linkbordercolor=red,   
    citebordercolor=green, 
    urlbordercolor=blue    
}
\usepackage{orcidlink}

\newtheorem{theorem}{Theorem}
\newtheorem{lemma}[theorem]{Lemma}
\newtheorem{corollary}[theorem]{Corollary}
\newtheorem{definition}[theorem]{Definition}
\newtheorem{proposition}[theorem]{Proposition}
\theoremstyle{definition}
\newtheorem{example}[theorem]{Example}

\newtheorem{notation}[theorem]{Notation}
\newtheorem{remark}[theorem]{Remark}

\newcommand{\prt}[1]{\left(#1\right)}
\newcommand{\corch}[1]{\left[#1\right]}
\newcommand{\claud}[1]{\left\{#1\right\}}
\newcommand{\tra}[1]{\operatorname{Tr}\left(#1\right)}
\newcommand{\ptr}[2]{\operatorname{Tr}_{#1}\left(#2\right)}
 \newcommand{\wt}[1]{w\left(#1\right)}
\newcommand{\dimwt}[1]{w_{\dim}\left(#1\right)}
\newcommand{\supp}[1]{\operatorname{supp}\left(#1\right)}

\begin{document}
\preprint{Published in Phys. Rev. A \textbf{114}, 022464 (2026). DOI: \href{https://doi.org/10.1103/fyq2-mpmm}{10.1103/fyq2-mpmm}}
\title{Mixed-dimensional quantum MacWilliams identity: Bounds for codes and absolutely maximally entangled states in heterogeneous systems}
\author{David González-Lociga \orcidlink{0009-0009-6991-3960}}
\email{davidgonzalezlociga@gmail.com}
\affiliation{Department of Physics, Universitat de Barcelona, 08028 Barcelona}

\author{Simeon Ball \orcidlink{0000-0003-4845-2084}}
\email{simeon.michael.ball@upc.edu}
\affiliation{Department of Mathematics, Universitat Politècnica de Catalunya, 08028 Barcelona}
\date{\today}
\begin{abstract}
As emerging quantum architectures evolve into heterogeneous networks combining different physical substrates, such as qubits for logic and higher-dimensional qudits for robust communication, the traditional scalar metrics of quantum error correction become insufficient. To address this, we introduce a mathematical framework based on dimension multisets to characterize quantum error-correcting codes (QECC) and absolutely maximally entangled (AME) states in mixed-dimensional Hilbert spaces. By replacing scalar weights with multisets, we accurately capture the exact physical composition of error supports across these diverse systems. Our central result is the mixed-dimensional quantum MacWilliams identity, which establishes the formal algebraic relationship between Shor-Laflamme enumerators and unitary weight enumerators. From this foundation, we deduce the mixed-dimensional shadow identity and derive rigorous, generalized constraints on code parameters, explicitly formulating the mixed-dimensional quantum Hamming, Singleton and Scott bounds, and developing a linear program to systematically evaluate code viability. For the Singleton bound, a tighter bound that has no homogeneous analogue is derived for pure mixed-dimensional codes. Finally, we deploy this enumerator machinery to thoroughly analyze AME states, utilizing shadow inequalities to constrain their existence and introducing a combinatorial grid method for the explicit construction of mixed-dimensional tripartite AME states.

\vspace{0.5cm}
\noindent 
DOI: \href{https://doi.org/10.1103/fyq2-mpmm}{10.1103/fyq2-mpmm}
\end{abstract}
\maketitle
\section{Introduction}

The evolution of quantum information processing from theoretical abstraction to physical realization increasingly demands mixed-dimensional architectures. Scalable quantum networks and fault-tolerant computers will likely rely on mixed-dimensional Hilbert spaces, integrating qubits for rapid logic operations with higher-dimensional qudits for robust memory and secure communication. Yet, the traditional mathematical machinery of quantum error correction, developed in the late 1990s and early 2000s by Gottesman, Calderbank, Rains, Sloane, Shor, and others \cite{Gottesman1997,CalderbankRainsShorSloane98}, is fundamentally rooted in homogeneous systems, assuming a constant local dimension across all subsystems. When applied to a mixed-dimensional Hilbert space, these conventional metrics fail. Specifically, the traditional scalar weight of an error, which assigns equal cost to any non-trivial subsystem alteration regardless of its local dimension, becomes an insufficient measure of error complexity.

The foundation for evaluating quantum code weight distributions was laid by Shor and Laflamme, who introduced two weight enumerators and derived a quantum analogue of the MacWilliams identities from classical coding theory \cite{ShorLaflamme97}. Rains later expanded this by defining two new unitary enumerators connected by a simpler transform, clarifying the constraints on Shor-Laflamme enumerators, determining the quantum MacWilliams identity for local dimension greater than 2 and proving a conjecture regarding shadow inequalities using polynomial invariants \cite{Rains1996, Rains00}. Together, these tools allowed for the formulation of linear programs to systematically rule out certain code parameters. Beyond error correction, this enumerator machinery proved highly effective in the study of absolutely maximally entangled (AME) states and pure codes. Scott utilized these tools to establish fundamental bounds for AME states \cite{Scott2004}, while Huber \textit{et al.} \cite{HuberEltschkaSiewertGuhne2018,HuberGrassl2020} leveraged shadow enumerators to prove the non-existence of large families of pure codes and AME states in homogeneous systems. Huber \textit{et al.} also hinted at the necessity of mixed-dimensional frameworks by providing an example of an AME state formed by a qubit and three qutrits in \cite{HuberEltschkaSiewertGuhne2018}. Recently, Ball and Zhang \cite{BallZhang2026} expanded this heterogeneous theory by formally defining a stabilizer formalism and AME states in mixed-dimensional Hilbert spaces, and proving the quantum Singleton bound for these systems, while also providing explicit examples across a wide variety of configurations. 

Other constructions of mixed-dimensional AME states are known. Goyeneche \textit{et al.} introduced orthogonal arrays to construct $k$-uniform states in homogeneous systems of $n$ parties and mixed orthogonal arrays for heterogeneous ones \cite{GoyenecheAlsinaLatorreRieraZyczkowski2015,GoyenecheBielawskiZyczkowski2016}. These are states whose reduction into $k$ subsystems is maximally mixed and if $k=\lfloor \frac{n}{2}\rfloor$, they are defined as absolutely maximally entangled states. Here we will follow a more general approach, as the work in \cite{HuberEltschkaSiewertGuhne2018} and \cite{BallZhang2026}. We will define a state as absolutely maximally entangled if any reduction onto a subsystem of dimension less or equal than the square root of the dimension of the total Hilbert space is maximally mixed. According to our definition, AME states are not necessarily $\lfloor \frac{n}{2}\rfloor$-uniform, but the converse is true. There also exist constructions of codes over mixed-dimensional Hilbert spaces such as the one in \cite{WangYuFanOh2013} with projection-based and graphical constructions.

In this work, we address the mathematical gap in mixed-dimensional quantum error correction and AME state characterization by introducing a framework based on dimension multisets. By replacing scalar weights with multisets, we accurately capture the exact physical composition of error supports. This fine-grained tracking allows us to generalize the standard enumerator polynomials into a multivariate form that naturally reduces to the constant-dimension case via a simple substitution. This multiset approach has recently proven to be useful in classical coding theory and more specifically for linear codes over finite commutative Frobenius rings \cite{Tang2025}.

Our central result is the proof of the mixed-dimensional quantum MacWilliams identity, which establishes the formal algebraic relationship between the multivariate Shor-Laflamme and unitary enumerators presented in this article. We must remark that a previous quantum MacWilliams identity for mixed-dimensional systems was derived in \cite{Nebe2006}; however, our multiset formulation significantly reduces the polynomial complexity by grouping systems of equal dimension and reduces perfectly to the homogeneous case. Building on this foundation, we derive the mixed-dimensional shadow identity and explicitly formulate generalized constraints on code parameters, including the mixed-dimensional quantum Hamming, Singleton and Scott bounds. In the case of the Singleton bound, a tighter bound that has no direct analogue in the homogeneous case is found for mixed-dimensional pure codes. We also construct a linear program capable of systematically evaluating the viability of heterogeneous code parameters. Finally, we apply this multiset machinery to the study of mixed-dimensional AME states, utilizing generalized shadow inequalities to heavily constrain their existence and we introduce a combinatorial grid methodology for the explicit construction of mixed-dimensional tripartite AME states similar to \cite{ShenChen21}. The complete code for our framework can be found in our GitHub repository \cite{gonzalezlociga_macwilliams_code}.

This article is organized as follows. In Section \ref{sec:mixed-dimensional_QECCS_and_Multisets}, we establish the notation for mixed-dimensional Hilbert spaces, redefine quantum error-correcting codes (QECCs) using dimensional weight and introduce the concept of dimension multisets alongside our novel multivariate weight enumerators. In Section \ref{sec:mixed-dimensional_macwilliams}, we state and prove the mixed-dimensional quantum MacWilliams identity. Section \ref{sec:enumerator_identities} establishes the mixed-dimensional shadow identity and derives linear transformations regarding the different weight enumerators previously introduced as an application of the MacWilliams and shadow identities. In Section \ref{sec:bounds_on_mixed-dimensional_codes}, we use this framework to prove generalized geometric and algebraic limits on codes, specifically the mixed-dimensional quantum Hamming and Singleton bounds, and we cast the enumerator constraints into a linear program. Finally, Section \ref{sec:ame_states} applies the theory to absolutely maximally entangled states, detailing the generalized Scott and shadow bounds, mapping the existence landscape of qubit-qutrit AME states and providing a rigorous combinatorial method for explicitly constructing mixed-dimensional tripartite AME states.
\section{Mixed-dimensional QECCs and multisets}\label{sec:mixed-dimensional_QECCS_and_Multisets}
To formally capture the physics of heterogeneous systems, we must first generalize the standard concepts of error supports and weights.

Let $\mathbb{H}=\bigotimes_{i=1}^{n} \mathbb{C}^{D_i}$ be a Hilbert space. We denote by $[n]$ the set $\{1,\dots,n\}$. Consider $S\subseteq[n]$ a subsystem with Hilbert space $\bigotimes_{i\in S} \mathbb{C}^{D_i}\subseteq \mathbb{H}$, then its dimension is $\dim S=\prod_{i\in S} {D_i}$. Naturally, the total dimension of the system is $\dim\mathbb{H}=\dim[n]=\prod_{i=1}^n D_i$ and for any bipartition, $\dim[n]=\dim{(S\cup S^c)}=\dim{S}\dim{S^c}$,  where $S^c$ is the complement of $S$.

For each subsystem $i \in [n]$, we consider an orthogonal operator basis $\mathcal{E}_i$ on $\mathbb{C}^{D_i}$ that includes the local identity. The orthogonality guarantees that all other local basis elements are traceless. We then construct a \textit{nice error basis} $\mathcal{E}$ acting on the full Hilbert space $\mathbb{H}$ \cite{Ketkar06} by taking tensor products of these local elements. That is,
$$\mathcal{E}=\bigotimes_{i=1}^n \mathcal{E}_i.$$
Consequently, every global error $E \in \mathcal{E}$ naturally decomposes as $E = \bigotimes_{i=1}^n E_i$, where each local non-identity operator $E_i$ is traceless. By construction, $\mathcal{E}$ contains the global identity matrix, is closed under multiplication up to a phase, and its elements satisfy the global orthogonality relation 
\begin{equation}\label{eq:orthogonality}
\tra{E_\alpha^\dagger E_\beta}=\delta_{\alpha,\beta}\dim[n]
\end{equation}
for $E_\alpha,E_\beta\in \mathcal{E}$. The generalized Pauli (Weyl-Heisenberg) operators serve as a standard example \cite{BallZhang2026}. Accordingly, $M$ an operator on $\mathbb{H}$ can be decomposed as 
\begin{equation}\label{eq:bloch} 
    M=(\dim[n])^{-1}\sum_{E\in \mathcal{E}}\tra{E^\dagger M}E,
\end{equation}
in the so-called \textit{Bloch representation}.

Let $E\in \mathcal{E}$ be a local operator on $\mathbb{H}$. The \textit{support} of $E$ is defined as the minimum-size subsystem $S\subseteq [n]$ such that $E=E_S\otimes \mathds{1}_{S^c}$. This will be denoted as $S=\supp{E}$. It is physically intuitive to think of the support as the parties in which the error acts non-trivially.

In the homogeneous literature, when all the dimensions are equal, we usually keep track of the \textit{weight} of an error $E$, which is defined as $\wt{E}=|\supp{E}|$, the cardinality of its support. However, in our mixed-dimensional context we not only need to keep track of how many systems the error acts upon but on which of them it does. Following this line of thought, the dimensional weight of an error was introduced in \cite{BallZhang2026}. For an error $E\in \mathcal{E}$ with support $S$, the \textit{dimensional weight} is defined as $\dimwt{E}=\dim S$. Notice that for constant local dimension $q$, the dimensional weight of an error $E$ with support $S$ reduces to $\dimwt{E}=q^{|S|}=q^{\wt{E}}$.

In light of the above, now we can properly define a quantum error-correcting code in our mixed-dimensional context.

A \textit{quantum error-correcting code} (QECC) with parameters $(((D_1,\dots, D_n),K,D))$ is a subspace $\mathcal{Q}\subseteq \mathbb{H}$ with dimension $K$. For any orthonormal basis $\{\ket{i_\mathcal{Q}}\}$ of $\mathcal{Q}$, the code must satisfy the Knill-Laflamme conditions \cite{KnillLaflamme97}
\begin{equation}\label{eq:knill-laflamme_conditions}
\bra{i_\mathcal{Q}}E\ket{j_\mathcal{Q}}=\delta_{ij}c_E,
\end{equation}
for all errors $E\in \mathcal{E}$ such that $\dimwt{E}<D$. We call $D$ the \textit{dimensional minimum distance}, which serves as the mixed-dimensional analogue to the minimum distance in standard literature. By replacing the conventional scalar weight with the dimensional weight, the typical threshold condition is naturally generalized. 

In Eq. \eqref{eq:knill-laflamme_conditions}, the coefficient $c_E$ depends solely on the error $E$ and not on the chosen basis vectors. If $c_E=\frac{\tra{E}}{\dim[n]}$, the QECC is called \textit{pure}. Since we have chosen $\mathcal{E}$ to be a nice error basis, the non-trivial errors are traceless; therefore, $c_E=0$ for all $E\neq \mathds{1}$ in a pure code. 

Errors satisfying $\dimwt{E}<D$ are formally classified as \textit{detectable errors}. In contrast, \textit{correctable errors} are those belonging to a specific subset of errors for which the code can successfully identify and reverse their action. For a code to correct a given set of errors, the Knill-Laflamme conditions must hold for the product of any two errors within that set. That is, for any two correctable errors $E_a$ and $E_b$, the combined operator $E_a^\dagger E_b$ must behave as a detectable error, satisfying $\dimwt{E_a^\dagger E_b} < D$. Because the dimensional weight of a combined error acting on disjoint supports is bounded by the product of their individual dimensional weights ($\dimwt{E_a^\dagger E_b} \le \dimwt{E_a} \dimwt{E_b}$), the code is guaranteed to correct all errors up to a dimensional weight threshold $T$ provided that $T^2 < D$.

While the Knill-Laflamme conditions dictate whether a code can correct a specific set of errors, fully characterizing a code's structure and deriving fundamental bounds requires analyzing its weight distribution. In homogeneous codes, the standard Shor-Laflamme enumerators group errors by a single integer: their scalar weight. However, in a mixed-dimensional setting, this scalar approach is fundamentally inadequate. For example, an error acting on a single ququart and an error acting on two qubits both possess a dimensional weight of 4, yet they have entirely different combinatorial footprints. We require another layer of fine-graining, the dimensional weight is not enough either.

To construct a valid enumerator framework for mixed dimensions, we must track not just the total size of an error's support, but its exact composition. We achieve this by introducing dimension multisets.

\begin{definition}Let $\mathbb{H}=\bigotimes_{i=1}^{n} \mathbb{C}^{D_i}$ be a mixed-dimensional Hilbert space, where $[n]=\{1,\dots,n\}$ is the set of all systems and $D_i$ is the local dimension of the $i$-th subsystem. For any subsystem $S\subseteq[n]$, we define the dimension multiset $\mathcal{D}(S)$ as the multiset of the local dimensions of the subsystems contained in $S$. That is
\begin{equation*}
\mathcal{D}(S)=\{D_i \,|\, i\in S\}.
\end{equation*}
For an error operator $E\in\mathcal{E}$, its dimension multiset is simply the dimension multiset of its support, i.e., $\mathcal{D}(\supp{E})$. 
\end{definition} 
Notice that the dimension multiset strictly generalizes both previous weight metrics. Given the dimension multiset of an error $E$, we can recover its standard scalar weight as $\wt{E}=|\mathcal{D}(\supp{E})|$ and its dimensional weight as $\dimwt{E}=\prod_{d \in \mathcal{D}(\supp{E})}d$. Similarly, for any subsystem $S$, its total physical dimension is straightforwardly computed as $\dim S= \prod_{d\in \mathcal{D}(S)}d.$

Before proceeding to further definitions, we establish the following notation.
\begin{notation} We fix:
    \begin{itemize}
        \item Boldface letters (e.g., $\mathbf{u}, \mathbf{v}, \mathbf{w}$) will be used exclusively to denote dimension multisets.
        \item Let $\mathbf{N}=\mathcal{D}([n])$ denote the dimension multiset of the entire space.
        \item Let $\mathbb{D}=\{d_1,d_2,\dots, d_{|\mathbb{D}|}\}$ be the set of distinct local dimensions present in $\mathbb{H}$.
        \item Let $m_{d}(\mathbf{u})$ denote the multiplicity of $d\in \mathbb{D}$ in a given dimension multiset $\mathbf{u}$.
        \item We say $\mathbf{u}$ is a \textit{sub-multiset} of $\mathbf{N}$, denoted as $\mathbf{u}\subseteq \mathbf{N}$, if the multiplicity of every element in $\mathbf{u}$ is less than or equal to its multiplicity in $\mathbf{N}$, meaning $m_d(\mathbf{u}) \le m_d(\mathbf{N})$ for all $d \in \mathbb{D}$.
        \item A multiset $\mathbf{v}$ has dimension $\dim \mathbf{v}=\prod_{d\in \mathbf{v}}d=\prod_{d\in\mathbb{D}}d^{m_d(\mathbf{v})}$.
    \end{itemize}
\end{notation}
From the standard enumerator machinery the only ones that admit the same formulation are the unitary calligraphic ones \cite{Rains1996}.
\begin{definition}\label{def:calligraphic_coeff}
    Given any two Hermitian operators $M$ and $N$ acting on $\mathbb{H}=\bigotimes_{i=1}^n \mathbb{C}^{D_i}$, the calligraphic coefficients for $S\subseteq [n]$ are defined as
    \begin{align*}
        \mathcal{A}'_S(M,N)&=\ptr{S}{\ptr{S^c}{M}\ptr{S^c}{N}},\\
        \mathcal{B}'_S(M,N)&=\ptr{S^c}{\ptr{S}{M}\ptr{S}{N}}.
    \end{align*}
\end{definition}
The others must be redefined, as follows.
\begin{definition}\label{def:ShorLaflamme_coeff}
     Given any two Hermitian operators $M$ and $N$ acting on $\mathbb{H}=\bigotimes_{i=1}^n \mathbb{C}^{D_i}$, the primal and dual Shor-Laflamme multiset coefficients for $\mathbf{v}\subseteq \mathbf{N}$ are defined as 
\begin{align*}
    A_{\mathbf{v}}(M,N)&=\sum_{\mathcal{D}(\supp{E})=\mathbf{v}} \tra{EM}\tra{E^\dagger N},\\
    B_{\mathbf{v}}(M,N)&=\sum_{\mathcal{D}(\supp{E})=\mathbf{v}} \tra{EME^\dagger N}.
\end{align*}
\end{definition}
\begin{definition} \label{def:unitary_coeff}
     The unitary multiset coefficients for $\mathbf{w}\subseteq \mathbf{N}$ are defined as
\begin{align*} 
    A'_{\mathbf{w}}(M,N)&=\sum_{\mathcal{D}(S)=\mathbf{w}}\ptr{S}{\ptr{S^c}{M}\ptr{S^c}{N}},\\
    B'_{\mathbf{w}}(M,N)&=\sum_{\mathcal{D}(S)=\mathbf{w}}\ptr{S^c}{\ptr{S}{M}\ptr{S}{N}}.
\end{align*}
\end{definition}

The new coefficients allow us to redefine our new weight enumerators.
\begin{definition}\label{def:ShorLaflamme_unitary_weight_enumerators}
    Let $\mathbb{H}=\bigotimes_{i=1}^n\mathbb{C}^{D_i}$ be a mixed-dimensional Hilbert space defined over the index set $[n]$. We associate a pair of variables $(x_d,y_d)$ to each distinct local dimension $d\in \mathbb{D}$ and group them into $\vec{x}=\{x_d\}_{d\in \mathbb{D}}$ and $\vec{y}=\{y_d\}_{d\in \mathbb{D}}$. For any two Hermitian operators $M$ and $N$ acting on $\mathbb{H}$, the multivariate Shor-Laflamme weight enumerators are
    \begin{align*}
        A_{MN}(\vec{x},\vec{y})&=\sum_{\mathbf{v}\subseteq \mathbf{N}}A_{\mathbf{v}}(M,N)\prod_{d\in \mathbb{D}} x_d^{m_d(\mathbf{N})-m_d(\mathbf{v})}y_d^{m_d(\mathbf{v})},\\
        B_{MN}(\vec{x},\vec{y})&=\sum_{\mathbf{v}\subseteq \mathbf{N}}B_{\mathbf{v}}(M,N)\prod_{d\in \mathbb{D}} x_d^{m_d(\mathbf{N})-m_d(\mathbf{v})}y_d^{m_d(\mathbf{v})}.
    \end{align*}
    Analogously, the multivariate unitary weight enumerators are defined as 
    \begin{align*}
        A'_{MN}(\vec{x},\vec{y})&=\sum_{\mathbf{w}\subseteq \mathbf{N}}A'_{\mathbf{w}}(M,N)\prod_{d\in \mathbb{D}} x_d^{m_d(\mathbf{N})-m_d(\mathbf{w})}y_d^{m_d(\mathbf{w})},\\
        B'_{MN}(\vec{x},\vec{y})&=\sum_{\mathbf{w}\subseteq \mathbf{N}}B'_{\mathbf{w}}(M,N)\prod_{d\in \mathbb{D}} x_d^{m_d(\mathbf{N})-m_d(\mathbf{w})}y_d^{m_d(\mathbf{w})}.
    \end{align*}
\end{definition}
This multivariate formulation is a strict generalization of the standard weight enumerator polynomials. In a homogeneous quantum code where all qudits share a constant dimension $q$, the set of distinct dimensions reduces to a single element, $\mathbb{D} = \{q\}$. Consequently, the set of formal variables collapses to a single pair, $(x, y)$. The dimension multisets are then uniquely characterized entirely by their cardinalities: the total system size $n = |\mathbf{N}|$, the standard error weight $j = |\mathbf{v}|$ and the subsystem size $m = |\mathbf{w}|$. The product over $d \in \mathbb{D}$ naturally disappears and by grouping all error multisets $\mathbf{v}$ of the same weight $j$, the sum over multisets simplifies to a sum over integers. This perfectly recovers the classical Shor-Laflamme weight enumerator, $A_{MN}(x,y) = \sum_{j=0}^n A_j(M,N) x^{n-j} y^j$ and the unitary weight enumerator, $A'_{MN}=\sum_{j=0}^n A'_j(M,N) x^{n-j} y^j$.

To clarify the construction of these multivariate polynomials, consider a minimal mixed-dimensional system.

\begin{example}
    Let the total system consist of one qubit and three qutrits, namely $\mathbb{H}=\mathbb{C}^2\otimes \prt{\mathbb{C}^3}^{\otimes 3}$ giving a total dimension multiset $\mathbf{N}=\{2,3,3,3\}$. The set of distinct dimensions is $\mathbb{D}=\{2,3\}$, which associates the formal variables $x_2, y_2$ for the qubit and $x_3, y_3$ for the qutrits. 
    
    By summing over all possible error sub-multisets $\mathbf{v} \subseteq \mathbf{N}$ and tracking their multiplicities $m_d(\mathbf{v})$, the multivariate Shor-Laflamme weight enumerators expand explicitly as
    \begin{align*}
        A(\vec{x},\vec{y}) &= A_{\emptyset} x_2 x_3^3 + A_{\{2\}} x_3^3y_2 + A_{\{3\}} x_2 x_3^2 y_3 + A_{\{2, 3\}}  x_3^2y_2 y_3 \\
        &\quad + A_{\{3, 3\}} x_2 x_3 y_3^2 + A_{\{2, 3, 3\}}  x_3 y_2y_3^2 + A_{\{3, 3, 3\}} x_2 y_3^3 + A_{\mathbf{N}} y_2 y_3^3,\\[2ex]
        B(\vec{x},\vec{y}) &= B_{\emptyset} x_2 x_3^3 + B_{\{2\}} x_3^3y_2 + B_{\{3\}} x_2 x_3^2 y_3 + B_{\{2, 3\}}  x_3^2y_2 y_3 \\
        &\quad + B_{\{3, 3\}} x_2 x_3 y_3^2 + B_{\{2, 3, 3\}}  x_3 y_2y_3^2 + B_{\{3, 3, 3\}} x_2 y_3^3 + B_{\mathbf{N}} y_2 y_3^3.\\[2ex]
    \end{align*}
\end{example}
When applying this enumerator machinery to a QECC $\mathcal{Q}$ of rank $K$, the matrices $M$ and $N$ are taken as the projector onto the code, that is, 
\begin{equation*}
M=N=\Pi_\mathcal{Q}=\sum_{i=1}^K\ketbra{i_\mathcal{Q}}{i_\mathcal{Q}},
\end{equation*}
for any orthonormal basis $\{\ket{i_\mathcal{Q}}\}$ of $\mathcal{Q}$. Therefore, we can define
\begin{equation*}
A_\mathbf{v}(\Pi_\mathcal{Q},\Pi_\mathcal{Q})\coloneqq A_\mathbf{v}(\Pi_\mathcal{Q}).
\end{equation*}
When the code the coefficients refer to is clear, we will drop the $\Pi_\mathcal{Q}$ notation. 

This particular coefficients enjoy the following properties. We follow the proof in \cite{HuberEltschkaSiewertGuhne2018} but for our multiset approach.
\begin{proposition}\label{prop:coeff_properties}
    Let $\mathcal{Q}\subset\mathbb{H}$ be QECC in a mixed-dimensional Hilbert space. Then, the Shor-Laflamme multiset coefficients satisfy
    \begin{align*}
        KB_\mathbf{v}(\Pi_\mathcal{Q})&\ge A_\mathbf{v}(\Pi_\mathcal{Q}) \ge 0,\\
        KB_\emptyset(\Pi_\mathcal{Q})&=A_\emptyset(\Pi_\mathcal{Q}),
    \end{align*}
    for all multisets $\mathbf{v}\subseteq\mathbf{N}$.
    Furthermore, any projector $\Pi_\mathcal{Q}$ of rank $K$ is a QECC of dimensional minimum distance $D$ if and only if
    \begin{equation*}
        KB_\mathbf{v}(\Pi_\mathcal{Q})= A_\mathbf{v}(\Pi_\mathcal{Q})
    \end{equation*}
    for all multisets $\mathbf{v}\subseteq\mathbf{N}$ such that $\dim \mathbf{v}<D$.
\end{proposition}

\begin{proof}
    Considering $M=N=\Pi_\mathcal{Q}=\sum_{i=1}^K\ketbra{i_\mathcal{Q}}{i_\mathcal{Q}}$ in Definition \ref{def:ShorLaflamme_coeff} we obtain
    \begin{align*}
        A_{\mathbf{v}}(\Pi_\mathcal{Q})&=\sum_{\mathcal{D}(\supp{E})=\mathbf{v}} \tra{E\Pi_\mathcal{Q}}\tra{E^\dagger \Pi_\mathcal{Q}}=\sum_{\mathcal{D}(\supp{E})=\mathbf{v}}\left|\sum_{i=1}^K\bra{i_\mathcal{Q}}E\ket{i_\mathcal{Q}}\right|^2,\\
        B_{\mathbf{v}}(\Pi_\mathcal{Q})&=\sum_{\mathcal{D}(\supp{E})=\mathbf{v}} \tra{E\Pi_\mathcal{Q}E^\dagger \Pi_\mathcal{Q}}=\sum_{\mathcal{D}(\supp{E})=\mathbf{v}} \sum_{i,j=1}^K \left|\bra{i_\mathcal{Q}}E\ket{j_\mathcal{Q}}\right|^2.
    \end{align*}

    Since $A_{\mathbf{v}}$ is a sum of absolute squares, we immediately have $A_{\mathbf{v}}(\Pi_\mathcal{Q}) \ge 0$. By applying the Cauchy-Schwarz inequality, we can relate the terms inside the sums for any error $E$
    \begin{equation*}
        \left|\sum_{i=1}^K\bra{i_\mathcal{Q}}E\ket{i_\mathcal{Q}}\right|^2 \le K \sum_{i=1}^K \left|\bra{i_\mathcal{Q}}E\ket{i_\mathcal{Q}}\right|^2 \le K \sum_{i,j=1}^K \left|\bra{i_\mathcal{Q}}E\ket{j_\mathcal{Q}}\right|^2.
    \end{equation*}
    Summing over all errors with $\mathcal{D}(\supp{E})=\mathbf{v}$ yields the first general inequality
    \begin{equation*}
        A_{\mathbf{v}}(\Pi_\mathcal{Q}) \le K B_{\mathbf{v}}(\Pi_\mathcal{Q}).
    \end{equation*}

    For the empty multiset $\mathbf{v}=\emptyset$, the only error is the identity operator $E=\mathds{1}$. Thus,
    \begin{align*}
        A_{\emptyset}(\Pi_\mathcal{Q}) &= |\tra{\mathds{1}\Pi_\mathcal{Q}}|^2 = K^2, \\
        B_{\emptyset}(\Pi_\mathcal{Q}) &= \tra{\mathds{1}\Pi_\mathcal{Q}\mathds{1}^\dagger \Pi_\mathcal{Q}} = \tra{\Pi_\mathcal{Q}} = K.
    \end{align*}
    Multiplying $B_{\emptyset}$ by $K$ directly yields $KB_{\emptyset}(\Pi_\mathcal{Q}) = A_{\emptyset}(\Pi_\mathcal{Q})$.

    To prove the final equivalence, assume first that $\Pi_\mathcal{Q}$ is a QECC of dimensional minimum distance $D$. Consider the multisets $\mathbf{v}$ such that $\dim \mathbf{v}<D$. From the Knill-Laflamme conditions in Eq. \eqref{eq:knill-laflamme_conditions} we get $\bra{i_\mathcal{Q}}E\ket{j_\mathcal{Q}}=\delta_{ij}c_E$, since the errors satisfy $\mathcal{D}(\supp{E})=\mathbf{v}$ and therefore $\dimwt{E}=\dim\mathbf{v}<D$. This means both coefficient equations reduce to 
    \begin{align*}
        A_{\mathbf{v}}(\Pi_\mathcal{Q})=\sum_{\mathcal{D}(\supp{E})=\mathbf{v}}K^2 c_Ec_E^*, \quad\text{ and }\quad
        B_{\mathbf{v}}(\Pi_\mathcal{Q})=\sum_{\mathcal{D}(\supp{E})=\mathbf{v}} Kc_Ec_E^*,
    \end{align*}
    which leads to $KB_\mathbf{v}(\Pi_\mathcal{Q})= A_\mathbf{v}(\Pi_\mathcal{Q})$.

    Conversely, assume $KB_\mathbf{v}(\Pi_\mathcal{Q})= A_\mathbf{v}(\Pi_\mathcal{Q})$ holds for all $\mathbf{v}$ with $\dim \mathbf{v} < D$. This strict equality implies that the Cauchy-Schwarz inequalities used earlier must be perfectly saturated for all errors $E$ with $\mathcal{D}(\supp{E})=\mathbf{v}$. 
    Saturation of the first inequality implies that the diagonal elements $\bra{i_\mathcal{Q}}E\ket{i_\mathcal{Q}}$ are constant and independent of $i$, let us call this constant $c_E$. 
    Saturation of the second inequality implies that all off-diagonal elements must vanish, meaning $\bra{i_\mathcal{Q}}E\ket{j_\mathcal{Q}} = 0$ for $i \neq j$. 
    Combining these two facts gives $\bra{i_\mathcal{Q}}E\ket{j_\mathcal{Q}} = \delta_{ij} c_E$, which are exactly the Knill-Laflamme conditions. Therefore, $\Pi_\mathcal{Q}$ is a valid QECC with dimensional minimum distance $D$.
\end{proof}
\section{The mixed-dimensional quantum MacWilliams identity}\label{sec:mixed-dimensional_macwilliams}
To prove the main theorem, we first establish two technical lemmas that evaluate the operator $\ptr{S^c}{M}\otimes \mathds{1}_{S^c}$ for $M$ an operator in $\mathbb{H}$ and $S\subseteq [n]$ in two different ways: via the Bloch representation and via a CPTP map. Following this, we introduce a combinatorial counting lemma, which together provide the necessary machinery to derive the mixed-dimensional MacWilliams identity.
\begin{lemma}\label{lemma:partial_trace_bloch}
    Given an operator $M$, its reduction onto a given subsystem $S\subseteq [n]$ tensored by the identity on the complement $S^c$ reads
    \begin{equation*}
    \ptr{S^c}{M}\otimes \mathds{1}_{S^c}=(\dim S)^{-1} \sum_{\supp{E}\subseteq S}\tra{E^\dagger M}E.
    \end{equation*}
\end{lemma}
\begin{proof} 
    By expressing $M$ in the Bloch representation following Eq. \eqref{eq:bloch} we obtain
    \begin{equation*}
    \ptr{S^c}{M}\otimes \mathds{1}_{S^c}=\ptr{S^c}{(\dim[n])^{-1}\sum_{E\in \mathcal{E}}\tra{E^\dagger M}E}\otimes \mathds{1}_{S^c}=(\dim[n])^{-1}\sum_{E\in\mathcal{E}}\tra{E^\dagger M} \ptr{S^c}{E}\otimes \mathds{1}_{S^c},
    \end{equation*}
    by linearity of the partial trace. Now we can partition our sum into two terms as
    \begin{equation}\label{eq:expanded_lemma_bloch} 
    \ptr{S^c}{M}\otimes \mathds{1}_{S^c}=(\dim[n])^{-1}\prt{\sum_{\supp{E}\subseteq S}\tra{E^\dagger M} \ptr{S^c}{E}+\sum_{\supp{E}\not\subseteq S}\tra{E^\dagger M} \ptr{S^c}{E}}\otimes \mathds{1}_{S^c}.
    \end{equation}

    Consider the second sum in Eq. \eqref{eq:expanded_lemma_bloch}. If $\supp{E}\not\subseteq S$, then $E$ contains at least one non-identity component on some subsystem of $S^c$. Due to the traceless property of the basis, $\ptr{S^c}{E}=0$.

    Now, if $\supp{E}\subseteq S$, $E$ can be expressed as the product of two operators acting on $S$ and $S^c$, namely, $E=E_S\otimes \mathds{1}_{S^c}$. Note $E$  acts trivially in the complement of $S$, hence $\ptr{S^c}{E}=\ptr{S^c}{E_S\otimes \mathds{1}_{S^c}}=E_S \tra{\mathds{1}_{S^c}}=E_S\dim{S^c}$. 

    \noindent Substituting in Eq. \eqref{eq:expanded_lemma_bloch} yields
    \begin{align*} 
    \ptr{S^c}{M}\otimes \mathds{1}_{S^c}&=(\dim[n])^{-1}\sum_{\supp{E}\subseteq S}\tra{E^\dagger M} \ptr{S^c}{E}\otimes \mathds{1}_{S^c}=\dim{S^c}(\dim[n])^{-1}\sum_{\supp{E}\subseteq S}\tra{E^\dagger M} E_S\otimes \mathds{1}_{S^c}\\&=\dim{S^c}(\dim[n])^{-1}\sum_{\supp{E}\subseteq S}\tra{E^\dagger M} E=(\dim{S})^{-1}\sum_{\supp{E}\subseteq S}\tra{E^\dagger M} E,
    \end{align*}
and the proof is finished.
\end{proof}
\begin{lemma}\label{lemma:partial_trace_CPTP}
Given an operator $M$, its partial trace over subsystem $S^c$ tensored by the identity on $S^c$ can be expressed as
\begin{equation*}
\ptr{S^c}{M}\otimes \mathds{1}_{S^c}=(\dim S^c)^{-1}\sum_{\supp{E}\subseteq S^c}EME^\dagger.
\end{equation*}
\end{lemma}
\begin{proof}
The sum $\sum EME^\dagger$ is invariant under a unitary change of the orthonormal operator basis. Consequently, we may prove the identity using the standard matrix basis $E_{ij} = \sqrt{\dim S^c}\ketbra{i}{j}$, where $\{\ket{i}\}$ is an orthonormal basis for $S^c$. We evaluate the sum
\begin{equation*}
    \sum_{\supp{E}\subseteq S^c}EME^\dagger =\sum_{i,j=1}^{\dim S^c} (\mathds{1}_{S} \otimes  \sqrt{\dim S^c}\ketbra{i}{j}) M (\mathds{1}_{S} \otimes \sqrt{\dim S^c}\ketbra{j}{i})=\prt{\dim S^c}\sum_{i,j=1}^{\dim S^c} (\mathds{1}_{S} \otimes  \ketbra{i}{j}) M (\mathds{1}_{S} \otimes \ketbra{j}{i})
\end{equation*}

Any operator $M$ acting on the total system can be uniquely decomposed by expanding its action on the subsystem $S^c$ in terms of these matrix units
\begin{equation*}
    M = \sum_{k,l=1}^{\dim S^c} M_{kl} \otimes \ketbra{k}{l},
\end{equation*}
where each $M_{kl}$ is an operator acting strictly on $S$. Substituting this decomposition into the expression above, we obtain
\begin{align*}
    \sum_{\supp{E}\subseteq S^c}EME^\dagger=& \prt{\dim S^c}\sum_{i,j} \sum_{k,l} (\mathds{1}_{S} \otimes \ketbra{i}{j}) (M_{kl} \otimes \ketbra{k}{l}) (\mathds{1}_{S} \otimes \ketbra{j}{i})= \prt{\dim S^c}\sum_{k,l} M_{kl} \otimes \left( \sum_{i,j} \ket{i}\braket{j}{k} \braket{l}{j}\bra{i} \right)\\
    &=\prt{\dim S^c}\sum_{k,l} M_{kl}\otimes \delta_{kl} \sum_{i} \ketbra{i}{i}=\prt{\dim S^c}\sum_{k}M_{kk}\otimes \mathds{1}_{S^c}= \prt{\dim S^c}\ptr{S^c}{M}\otimes \mathds{1}_{S^c},
\end{align*}
which rearranging finishes the proof.
\end{proof}

\begin{lemma}\label{lemma:counting_Cvw}
Let $\mathbb{H}=\bigotimes_{i=1}^n \mathbb{C}^{D_i}$ be a mixed-dimensional Hilbert space defined over the the index set $[n]$. 

For all errors $E\in\mathcal{E}$ with $K=\supp{E}$ and $\mathbf{v}=\mathcal{D}(K)$, the number of subsystems $S\subseteq [n]$ such that $K\subseteq S$ and $\mathcal{D}(S)=\mathbf{w}$ is given by
\begin{equation*}
C_{\mathbf{v},\mathbf{w}}\coloneqq \prod_{d\in \mathbb{D}}\binom{m_d(\mathbf{N})-m_{d}(\mathbf{v})}{m_d(\mathbf{w})-m_{d}(\mathbf{v})}.
\end{equation*}
\end{lemma}
\begin{proof}
    For every distinct dimension $d\in\mathbb{D}$, define $[n]_d=\{i\in [n] \,|\, D_i=d\}$ as the set of indices corresponding to the subsystems that have local dimension $d$. This partitions the total index space into a disjoint union: $[n]=\coprod_{d\in \mathbb{D}} [n]_d$. Analogously, we define the dimension-specific components of $K$ and $S$  by intersecting them with these partitioned subsets as
    \begin{equation*}
        K_d=K\cap [n]_d, \quad S_d=S\cap [n]_d.
    \end{equation*}
    This naturally induces a partition on $K$ and $S$. Given the dimension multiset constraints $\mathcal{D}(K)=\mathbf{v}$ and $\mathcal{D}(S)=\mathbf{w}$, for some $\mathbf{v,w}\subseteq \mathbf{N}$, we can relate the cardinality of each subset with the multiplicity of that dimension, i.e., 
    \begin{equation}\label{eq:cardinality=multiplicity}
        |K_d|=m_d(\mathbf{v}),\quad |S_d|=m_d(\mathbf{w}).
    \end{equation}
    We aim to construct subsystems $S$ such that $K\subseteq S\subseteq [n]$ and $\mathcal{D}(S)=\mathbf{w}$. Because the index partitions are mutually disjoint, this is equivalent to independently constructing a subset $S_d$ for every $d\in \mathbb{D}$ such that $K_d\subseteq S_d\subseteq [n]_d $ and $|S_d|=m_d(\mathbf{w})$. To form each $S_d$, we need to include all $|K_d|$  elements of $K_d$ and choose the remaining $|S_d|-|K_d|$ elements from $[n]_d\setminus K_d$, which has size $|[n]_d|-|K_d|$. The number of ways to make this choice is given by the binomial coefficient 
    \begin{equation*}
        \binom{|[n]_d|-|K_d|}{|S_d|-|K_d|}=\binom{m_d(\mathbf{N})-m_d(\mathbf{v})}{m_d(\mathbf{w})-m_d(\mathbf{v})},
    \end{equation*}
    where we have used Eq. \eqref{eq:cardinality=multiplicity}. Since the partitions for each dimension $d\in \mathbb{D}$ are strictly independent of one another, the total number of valid subsystems is the product of the independent choices across all dimensions, 
    \begin{equation*}
C_{\mathbf{v,w}}=\prod_{d\in \mathbb{D}}\binom{m_d(\mathbf{N})-m_d(\mathbf{v})}{m_d(\mathbf{w})-m_d(\mathbf{v})},
    \end{equation*}
    which finishes the proof.
\end{proof}
Notice that the support of the error has disappeared from the equation, meaning its dimension multiset is enough to count the number of subsystems fulfilling the constraints. Consequently, $C_{\mathbf{v,w}}$ is well-defined.

This result directly generalizes the combinatorial factors used in \cite{HuberEltschkaSiewertGuhne2018}. In the homogeneous case, the dimension multiset of an error becomes a uniform multiset whose cardinality represents the standard error weight, $j = |\mathbf{v}|$. Consequently, the profile constraint $\mathcal{D}(S)=\mathbf{w}$ in Lemma \ref{lemma:counting_Cvw} simply fixes the total size of the target subsystem to $m = |\mathbf{w}|$. In this limit, $C_{\mathbf{v,w}}$ evaluates the number of ways to form a subsystem of size $m$ that fully contains an error of weight $j$ in an $n$-qudit system, naturally recovering the familiar binomial coefficient $\binom{n-j}{m-j}$. More explicitly, the product over $d \in \mathbb{D}$ collapses to a single term and the multiset multiplicities map directly to the standard scalar parameters. The total number of available qudits $m_d(\mathbf{N})$ becomes the total system size $n$. The number of qudits in the error's support $m_d(\mathbf{v})$ becomes the standard error weight $j$. Finally, the number of qudits required for the target subsystem $m_d(\mathbf{w})$ becomes the subsystem size $m$. Substituting these directly into the generalized expression yields exactly $\binom{n-j}{m-j}=\binom{n-j}{n-m}$.

\begin{theorem}\label{thm:mixed-dimensionalMacWilliams}
    For any two Hermitian operators $M$ and $N$ acting on $\mathbb{H}=\bigotimes_{i=1}^n \mathbb{C}^{D_i}$, the following identity holds,
    \begin{equation}
    A_{MN}\prt{\claud{x_d}_{d\in \mathbb{D}},\claud{y_d}_{d\in \mathbb{D}}}=B_{MN}\prt{\claud{\frac{x_d+(d^2-1)y_d}{d}}_{d\in\mathbb{D}},\claud{\frac{x_d-y_d}{d}}_{d\in\mathbb{D}}}.
    \end{equation}
\end{theorem}
\begin{proof}
    Our goal is to relate the unitary weights defined in Definition \ref{def:unitary_coeff} with the ones in Definition \ref{def:ShorLaflamme_coeff}. To do so, we will express the partial trace of an operator in two ways: first as a Bloch decomposition using Lemma \ref{lemma:partial_trace_bloch} and then as a quantum channel using Lemma \ref{lemma:partial_trace_CPTP}.

    First, we evaluate the unitary coefficient $\mathcal{A}'_S(M,N)$ using the Bloch decomposition provided by Lemma \ref{lemma:partial_trace_bloch} as
    \begin{align*}
        \mathcal{A}'_{S}(M,N)&=\ptr{S}{\ptr{S^c}{M}\ptr{S^c}{N}}=\tra{\ptr{S^c}{M}\otimes \mathds{1}_{S^c}N}\\
        &=\tra{(\dim S)^{-1} \sum_{\supp{E}\subseteq S}\tra{E M}E^\dagger\, \corch{(\dim [n])^{-1}\sum_{F\in \mathcal{E}}\tra{F^\dagger N}F}}\\
        &=(\dim S)^{-1}(\dim [n])^{-1}\sum_{\supp{E}\subseteq S} \tra{E M}\sum_{F\in\mathcal{E}} \tra{F^\dagger N}\tra{E^\dagger F}\\&=(\dim S)^{-1}\sum_{\supp{E}\subseteq S} \tra{E M} \tra{E^\dagger N},
    \end{align*}
    where we have used the Bloch representation for $N$, Lemma \ref{lemma:partial_trace_bloch} and the orthogonality relations for the orthonormal basis of operators from Eq. \eqref{eq:orthogonality}. Now we sum over all subsystems $S$ such that $\mathcal{D}(S)=\mathbf{w}$, for a fixed dimension multiset $\mathbf{w}$. Notice this yields the exact definition of the unitary coefficients in Definition \ref{def:unitary_coeff}. Therefore,
    \begin{align*}
        A'_{\mathbf{w}}(M,N)&=\sum_{\mathcal{D}(S)=\mathbf{w}}(\dim S)^{-1}\sum_{\supp{E}\subseteq S} \tra{E M} \tra{E^\dagger N}\\&=\prt{\prod_{d\in \mathbf{w}}d^{-1}}\sum_{\mathcal{D}(S)=\mathbf{w}}\,\sum_{\supp{E}\subseteq S} \tra{EM}\tra{E^\dagger N}.
    \end{align*}
    Swapping the sums allows us to factor out the trace-terms, since they do not depend on subsystem $S$. This gives
    \begin{align*}
        A'_{\mathbf{w}}(M,N)&=\prt{\prod_{d\in \mathbf{w}}d^{-1}}\sum_{E\in \mathcal{E}}\tra{EM}\tra{E^\dagger N}\sum_{\substack{S\supseteq \supp{E}\\ \mathcal{D}(S)=\mathbf{w}}} 1.
    \end{align*}
    The innermost sum simply counts the number of valid subsystems that fully contain the support of the error and possess the target dimension multiset $\mathbf{w}$. By Lemma \ref{lemma:counting_Cvw}, this count is exactly $C_{\mathcal{D}(\supp{E}),\mathbf{w}}$. We can now partition the sum over all errors $E \in \mathcal{E}$ by grouping them according to their dimension multiset $\mathbf{v} = \mathcal{D}(\supp{E})$ as
\begin{align*}
    A'_{\mathbf{w}}(M,N)&=\prt{\prod_{d\in \mathbf{w}}d^{-1}}\sum_{E\in \mathcal{E}}C_{\mathcal{D}(\supp{E}),\mathbf{w}}\tra{EM}\tra{E^\dagger N}\\
    &=\prt{\prod_{d\in \mathbf{w}}d^{-1}}\sum_{\mathbf{v}\subseteq \mathbf{N}} \sum_{\substack{E\in\mathcal{E}\\\mathcal{D}(\supp{E})=\mathbf{v}}} C_{\mathbf{v,w}}\tra{EM}\tra{E^\dagger N}.
\end{align*}
Because the combinatorial coefficient $C_{\mathbf{v,w}}$ evaluates to zero unless $\mathbf{v} \subseteq \mathbf{w}$, the outer sum naturally restricts its bounds. Factoring the constant $C_{\mathbf{v,w}}$ out of the inner sum yields
\begin{equation}\label{eq:AdashW=sumAV}
    \begin{split}
    A'_{\mathbf{w}}(M,N)&=\prt{\prod_{d\in \mathbf{w}}d^{-1}}\sum_{\mathbf{v}\subseteq \mathbf{w}} C_{\mathbf{v,w}} \corch{\sum_{\mathcal{D}(\supp{E})=\mathbf{v}}\tra{EM}\tra{E^\dagger N}}=\prt{\prod_{d\in \mathbf{w}}d^{-1}}\sum_{\mathbf{v}\subseteq \mathbf{w}} C_{\mathbf{v,w}}A_{\mathbf{v}}(M,N).
    \end{split}
\end{equation}

Conversely, evaluating the unitary dual coefficient $\mathcal{B}'_S(M,N)$ using Lemma \ref{lemma:partial_trace_CPTP},
\begin{align*}
\mathcal{B}'_{S}(M,N)&=\ptr{S^c}{\ptr{S}{M}\otimes \ptr{S}{N}}=\tra{\ptr{S}{M}\otimes \mathds{1}_{S} N}\\
&=\tra{(\dim S)^{-1}\sum_{\supp{E}\subseteq S}EME^\dagger N}=(\dim S)^{-1}\sum_{\supp{E}\subseteq S}\tra{EME^\dagger N}.
\end{align*}
By summing over all subsystems $S$ such that $\mathcal{D}(S)=\mathbf{w}$ and proceeding analogously as before we obtain the $B'_{\mathbf{w}}(M,N)$ coefficients in terms of $B_{\mathbf{w}}(M,N)$ with the same combinatorial number. That is
\begin{align*}
    B'_{\mathbf{w}}(M,N)=\prt{\prod_{d\in \mathbf{w}}d^{-1}}\sum_{\mathbf{v}\subseteq \mathbf{w}} C_{\mathbf{v,w}}B_{\mathbf{v}}(M,N).
\end{align*}
As in the homogeneous scenario, the multivariate unitary weight enumerators are related to the multivariate Shor-Laflamme weight enumerators by means of a polynomial transform in each component.

By expanding the multivariate unitary weight enumerator and substituting the relation from Eq. \eqref{eq:AdashW=sumAV}, we obtain
\begin{align*}
A'_{MN}(\vec{x},\vec{y})&=\sum_{\mathbf{w}\subseteq \mathbf{N}}A'_{\mathbf{w}}(M,N)\prod_{d\in \mathbb{D}} x_d^{m_d(\mathbf{N})-m_d(\mathbf{w})}y_d^{m_d(\mathbf{w})}\\
&=\sum_{\mathbf{w}\subseteq \mathbf{N}}\corch{\prt{\prod_{d\in \mathbf{w}}d^{-1}}\sum_{\mathbf{v}\subseteq \mathbf{w}} C_{\mathbf{v,w}}A_{\mathbf{v}}(M,N)}\prod_{d\in \mathbb{D}} x_d^{m_d(\mathbf{N})-m_d(\mathbf{w})}y_d^{m_d(\mathbf{w})}.
\end{align*}
We can rewrite the prefactor $\prod_{d\in \mathbf{w}}d^{-1}$ in terms of the dimension multiplicities as $\prod_{d\in \mathbb{D}}d^{-m_d(\mathbf{w})}$. Absorbing this term into the formal variable $y_d$ yields
\begin{align*}
A'_{MN}(\vec{x},\vec{y})&=\sum_{\mathbf{w}\subseteq \mathbf{N}} \sum_{\mathbf{v}\subseteq \mathbf{w}}\prt{\prod_{d\in \mathbb{D}}d^{-m_d(\mathbf{w})}} C_{\mathbf{v,w}}A_{\mathbf{v}}(M,N)\prod_{d\in \mathbb{D}} x_d^{m_d(\mathbf{N})-m_d(\mathbf{w})}y_d^{m_d(\mathbf{w})}\\
&=\sum_{\mathbf{w}\subseteq \mathbf{N}} \sum_{\mathbf{v}\subseteq \mathbf{w}} C_{\mathbf{v,w}}A_{\mathbf{v}}(M,N)\prod_{d\in \mathbb{D}} x_d^{m_d(\mathbf{N})-m_d(\mathbf{w})}\prt{\frac{y_d}{d}}^{m_d(\mathbf{w})}.
\end{align*}
Next, we exchange the order of summation. Since $\mathbf{w} \subseteq \mathbf{N}$ and $\mathbf{v} \subseteq \mathbf{w}$, we can equivalently sum over all $\mathbf{v} \subseteq \mathbf{N}$ and then sum over all $\mathbf{w}$ such that $\mathbf{v} \subseteq \mathbf{w} \subseteq \mathbf{N}$. Pulling $A_{\mathbf{v}}(M,N)$ out of the inner sum gives
\begin{equation*}
A'_{MN}(\vec{x},\vec{y})=\sum_{\mathbf{v}\subseteq \mathbf{N}}  A_{\mathbf{v}}(M,N)\sum_{\mathbf{w}\supseteq \mathbf{v}} C_{\mathbf{v,w}}\prod_{d\in \mathbb{D}} x_d^{m_d(\mathbf{N})-m_d(\mathbf{w})}\prt{\frac{y_d}{d}}^{m_d(\mathbf{w})}.
\end{equation*}
Substituting the explicit combinatorial definition of $C_{\mathbf{v,w}}$ from Lemma \ref{lemma:counting_Cvw} and merging the products over $d \in \mathbb{D}$, we find
\begin{align*}
A'_{MN}(\vec{x},\vec{y})&=\sum_{\mathbf{v}\subseteq \mathbf{N}}  A_{\mathbf{v}}(M,N)\sum_{\mathbf{w}\supseteq \mathbf{v}} \prt{\prod_{d\in \mathbb{D}}\binom{m_d(\mathbf{N})-m_d(\mathbf{v})}{m_d(\mathbf{w})-m_d(\mathbf{v})}}\prod_{d\in \mathbb{D}} x_d^{m_d(\mathbf{N})-m_d(\mathbf{w})}\prt{\frac{y_d}{d}}^{m_d(\mathbf{w})}\\
&=\sum_{\mathbf{v}\subseteq \mathbf{N}}  A_{\mathbf{v}}(M,N)\sum_{\mathbf{w}\supseteq \mathbf{v}} \prod_{d\in \mathbb{D}} \binom{m_d(\mathbf{N})-m_d(\mathbf{v})}{m_d(\mathbf{w})-m_d(\mathbf{v})} x_d^{m_d(\mathbf{N})-m_d(\mathbf{w})}\prt{\frac{y_d}{d}}^{m_d(\mathbf{w})}.
\end{align*}
Because the constraints on $\mathbf{w}$ strictly decouple across each distinct local dimension $d$, the sum over the multisets $\mathbf{w} \supseteq \mathbf{v}$ is equivalent to summing independently over the integer multiplicities $m_d(\mathbf{w})$ from $m_d(\mathbf{v})$ up to $m_d(\mathbf{N})$. This allows us to exchange the sum and the product as
\begin{equation*}
A'_{MN}(\vec{x},\vec{y})=\sum_{\mathbf{v}\subseteq \mathbf{N}}  A_{\mathbf{v}}(M,N) \prod_{d\in \mathbb{D}} \corch{\sum_{m_d(\mathbf{w})=m_d(\mathbf{v})}^{m_d(\mathbf{N})}\binom{m_d(\mathbf{N})-m_d(\mathbf{v})}{m_d(\mathbf{w})-m_d(\mathbf{v})} x_d^{m_d(\mathbf{N})-m_d(\mathbf{w})}\prt{\frac{y_d}{d}}^{m_d(\mathbf{w})}}.
\end{equation*}
To evaluate the inner sum, we factor out $(y_d/d)^{m_d(\mathbf{v})}$ to match the standard form of the binomial theorem, shifting the exponent. Namely,
\begin{align*}
A'_{MN}(\vec{x},\vec{y})&=\sum_{\mathbf{v}\subseteq \mathbf{N}}  A_{\mathbf{v}}(M,N) \prod_{d\in \mathbb{D}} \prt{\frac{y_d}{d}}^{m_d(\mathbf{v})}\corch{\sum_{m_d(\mathbf{w})=m_d(\mathbf{v})}^{m_d(\mathbf{N})}\binom{m_d(\mathbf{N})-m_d(\mathbf{v})}{m_d(\mathbf{w})-m_d(\mathbf{v})} x_d^{m_d(\mathbf{N})-m_d(\mathbf{w})}\prt{\frac{y_d}{d}}^{m_d(\mathbf{w})-m_d(\mathbf{v})}}\\
&=\sum_{\mathbf{v}\subseteq \mathbf{N}}  A_{\mathbf{v}}(M,N) \prod_{d\in \mathbb{D}} \prt{x_d+\frac{y_d}{d}}^{m_d(\mathbf{N})-m_d(\mathbf{v})}\prt{\frac{y_d}{d}}^{m_d(\mathbf{v})}.
\end{align*}

Finally, recognizing this structure as the multiset Shor-Laflamme weight enumerator evaluated at a transformed set of variables, we arrive at the identity
\begin{equation}\label{eq:Adash=A_polynomial}
A'_{MN}\prt{\claud{x_d}_{d\in\mathbb{D}},\claud{y_d}_{d\in\mathbb{D}}}=A_{MN}\prt{\claud{x_d+\frac{y_d}{d}}_{d\in \mathbb{D}},\claud{\frac{y_d}{d}}_{d\in \mathbb{D}}}.
\end{equation}

Analogously, it is proved that 
\begin{equation}\label{eq:Bdash=B_polynomial}
    B'_{MN}\prt{\claud{x_d}_{d\in\mathbb{D}},\claud{y_d}_{d\in\mathbb{D}}}=B_{MN}\prt{\claud{x_d+\frac{y_d}{d}}_{d\in \mathbb{D}},\claud{\frac{y_d}{d}}_{d\in \mathbb{D}}}.
\end{equation}

Notice that by Definition \ref{def:calligraphic_coeff}, $\mathcal{B}'_{S}(M,N)=\mathcal{A}'_{S^c}(M,N)$, therefore $B'_{\mathbf{w}}(M,N)=A'_{\mathbf{N}\setminus \mathbf{w}}(M,N)$. This means that
\begin{align*}
    B'_{MN}(\claud{y_d}_{d\in\mathbb{D}},\claud{x_d}_{d\in\mathbb{D}})&=\sum_{\mathbf{w}\subseteq \mathbf{N}}B'_{\mathbf{w}}(M,N)\prod_{d\in \mathbb{D}} y_d^{m_d(\mathbf{N})-m_d(\mathbf{w})}x_d^{m_d(\mathbf{w})}\\&=\sum_{\mathbf{w}\subseteq \mathbf{N}}A'_{\mathbf{N}\setminus \mathbf{w}}(M,N)\prod_{d\in \mathbb{D}} y_d^{m_d(\mathbf{N})-m_d(\mathbf{w})}x_d^{m_d(\mathbf{w})}\\
    &=\sum_{\mathbf{w'}\subseteq \mathbf{N}}A'_{\mathbf{w'}}(M,N)\prod_{d\in \mathbb{D}} y_d^{m_d(\mathbf{w'})}x_d^{m_d(\mathbf{N})-m_d(\mathbf{w'})}= A'_{MN}(\claud{x_d}_{d\in\mathbb{D}},\claud{y_d}_{d\in\mathbb{D}})
\end{align*}
where we have defined $\mathbf{w'}=\mathbf{N}\setminus \mathbf{w}$. Therefore using Eqs. \eqref{eq:Adash=A_polynomial} and \eqref{eq:Bdash=B_polynomial},
\begin{align}
A_{MN}\prt{\claud{x_d}_{d\in\mathbb{D}},\claud{y_d}_{d\in\mathbb{D}}} &= A'_{MN}\prt{\claud{x_d - y_d}_{d\in\mathbb{D}}, \claud{d \cdot y_d}_{d\in\mathbb{D}}}\label{eq:A=Adash_polynomial}\\
&= B'_{MN}\prt{\claud{d \cdot y_d}_{d\in\mathbb{D}}, \claud{x_d - y_d}_{d\in\mathbb{D}}}\nonumber\\ &= B_{MN}\prt{\claud{\frac{x_d + (d^2 - 1)y_d}{d}}_{d\in\mathbb{D}}, \claud{\frac{x_d - y_d}{d}}_{d\in\mathbb{D}}}\nonumber,
\end{align}

and the proof is finished.
\end{proof}
\begin{remark}
We note that a previous mixed-dimensional MacWilliams identity was derived in \cite{Nebe2006}, by assigning a pair of variables $(x_i,y_i)$ to every individual subsystem $i\in \{1,\dots,n\}$. Such a formulation results in a polynomial containing $2^n$ terms, one for every subset, whereas our multiset approach yields 
\begin{equation*}
    \prod_{d\in \mathbb{D}}\prt{m_d(\mathbf{N})+1}
\end{equation*}
terms, which is equal to the number of possible sub-multisets $\mathbf{v}\subseteq\mathbf{N}$. Our approach reduces the number of terms considerably by compressing equal dimensional systems into the same variable. In fact, our worst-case scenario is when every single subsystem has a different dimension, meaning $m_d(\mathbf{N})=1$ for all $d\in \mathbb{D}$, recovering $2^n$ terms and matching the previous result. As an example, consider a system with 10 qubits and 10 qutrits. The previous number of terms would be $2^{20}=1,048,576$. With the multiset approach this reduces to $(10+1)\times(10+1)=121$ terms. Moreover, when considering constant-local dimensions in the previous approach, the standard homogeneous quantum MacWilliams identity with just two variables is not recovered.
\end{remark}
\section{Enumerator identities}\label{sec:enumerator_identities}
Having established the mixed-dimensional quantum MacWilliams identity, we now develop the algebraic tools necessary to bound code parameters. In this section, we derive the mixed-dimensional shadow identity and provide explicit inversion formulas that express the unitary, dual and shadow multiset coefficients as linear combinations of the primal Shor-Laflamme multiset coefficients.
\subsection{Shadow identity}
 It is known that for all positive semi-definite hermitian operators $M$ and $N$ on parties $\{1,\dots,n\}$ and any fixed subset $T\subseteq \{1,\dots, n\}$ it holds that \cite{Rains00}
 \begin{equation}\label{eq:generalized_shadow_inequality}
    \sum_{S\subseteq \{1,\dots,n\}} (-1)^{|S\cap T|} \ptr{S}{\ptr{S^c}{M}\ptr{S^c}{N}}\ge 0.
 \end{equation}
Note this is applicable to any number of parties and local dimensions. 
As before, we redefine the shadow coefficient following \cite{HuberEltschkaSiewertGuhne2018}.
\begin{definition}
     Given any two Hermitian operators $M$ and $N$ acting on $\mathbb{H}=\bigotimes_{i=1}^n \mathbb{C}^{D_i}$, the shadow multiset coefficient for $S\subseteq [n]$ are defined as
     \begin{equation}\label{eq:shadow_coeff}
        S_\mathbf{w}(M,N)=\sum_{\mathcal{D}(T)=\mathbf{w}} \sum_{S\subseteq [n]} (-1)^{|S\cap T^c|} \ptr{S}{\ptr{S^c}{M
        }\ptr{S^c}{N}},
     \end{equation}
     where $\mathbf{w}\subseteq \mathbf{N}$ is a fixed dimension multiset.
\end{definition}
Note the coefficient includes the complement of $T$ and not $T$ itself, following the procedure in \cite{HuberEltschkaSiewertGuhne2018}. A similar proof can be derived in terms of $T$, as done in \cite{HuberGrassl2020}.

By means of Eq. \eqref{eq:generalized_shadow_inequality}, every coefficient $S_{\mathbf{w}}(M,N)$ must be non-negative.
This can also be translated into a new polynomial.
\begin{definition}
The multiset shadow weight enumerator is the polynomial
\begin{equation*}
S_{MN}(\vec{x},\vec{y})=\sum_{\mathbf{w}\subseteq \mathbf{N}}S_{\mathbf{w}}(M,N)\prod_{d\in \mathbb{D}} x_d^{m_d(\mathbf{N})-m_d(\mathbf{w})}y_d^{m_d(\mathbf{w})},
\end{equation*}
where $\vec{x}=\claud{x_d}_{d\in\mathbb{D}}$ and $\vec{y}=\claud{y_d}_{d\in\mathbb{D}}$ as before.
\end{definition}
As in the homogeneous case, we can find a relation between the Shadow enumerator and the Shor-Laflamme one: the shadow identity.
\begin{theorem}\label{thm:shadow_identity}
    For any two Hermitian operators $M$ and $N$ acting on $\mathbb{H}=\bigotimes_{i=1}^n \mathbb{C}^{D_i}$, their multivariate shadow weight enumerator is related to their multivariate Shor-Laflamme weight enumerator by 
    \begin{equation*}
        S_{MN}\prt{\claud{x_d}_{d\in\mathbb{D}},\claud{y_d}_{d\in\mathbb{D}}}=A_{MN}\prt{\claud{\frac{(d-1)x_d+(d+1)y_d}{d}}_{d\in\mathbb{D}},\claud{\frac{y_d-x_d}{d}}_{d\in\mathbb{D}}}.
    \end{equation*}
\end{theorem}
\begin{proof}
    Recalling Eq. \eqref{eq:shadow_coeff} we can swap the sums and factor out the trace-terms, as usual. This yields 
    \begin{align*}
        S_\mathbf{w}(M,N)=\sum_{S\subseteq [n]}\ptr{S}{\ptr{S^c}{M}\ptr{S^c}{N}}\sum_{\mathcal{D}(T)=\mathbf{w}}(-1)^{|S\cap T^c|}.
    \end{align*}
    Following the logic from proof of Lemma \ref{lemma:counting_Cvw}, we partition our whole index space into the different dimensions. By doing so the cardinal becomes $|S\cap T^c|=|\coprod_{d\in \mathbb{D}} S_d\cap T^c_d|=\sum_{d\in \mathbb{D}}|S_d\cap T^c_d|$. This allows for the second sum in the equation above to be expressed as
    \begin{align*}
    \sum_{\mathcal{D}(T)=\mathbf{w}} (-1)^{|S\cap T^c|}=\sum_{\mathcal{D}(T)=\mathbf{w}} \prod_{d\in \mathbb{D}} (-1)^{|S_d\cap T^c_d|}.
    \end{align*}
    Notice the constraint $\mathcal{D}(T)=\mathbf{w}$ forces us to choose $m_d(\mathbf{w})$ elements for each subset $T_d$ from the $[n]_d$ available. This choice is independent for each dimension, so the sum and the product can be exchanged, yielding
    \begin{align*}
        \sum_{\mathcal{D}(T)=\mathbf{w}}\prod_{d\in \mathbb{D}}(-1)^{|S_d\cap T^c_d|}=\prod_{d\in \mathbb{D}}\corch{\sum_{\substack{T_d \subseteq [n]_d\\|T_d|=m_d(\mathbf{w})}} (-1)^{|S_d\cap T^c_d|}}.
    \end{align*}
    We are summing the signed overlaps between the fixed set $S_d$ and all possible subsets $T^c_d\subseteq [n]_d$ with fixed size $|T^c_d|=m_d(\mathbf{N})-m_d(\mathbf{w})$. Consider the dimension multiset of $S$, $\mathcal{D}(S)=\mathbf{v}$ and define $\alpha= |S_d\cap T^c_d|$. To construct a valid $T^c_d$ such that the overlap with $S_d$ is $\alpha$ and its size is $m_d(\mathbf{N})-m_d(\mathbf{w})$ we choose first $\alpha$ elements from those available in $S_d$, which are $m_d(\mathbf{v})$ and the remaining $m_d(\mathbf{N})-m_d(\mathbf{w})-\alpha$ from the $m_d(\mathbf{N})-m_d(\mathbf{v})$ available in $S^c_d$. Summing over all possible intersections $\alpha$, we end up with
    \begin{align*}
    \prod_{d\in \mathbb{D}}\corch{\sum_{\substack{T_d \subseteq [n]_d\\|T_d|=m_d(\mathbf{w})}} (-1)^{|S_d\cap T^c_d|}}=\prod_{d\in \mathbb{D}}\corch{\sum_{\alpha} (-1)^\alpha\binom{m_d(\mathbf{N})-m_d(\mathbf{v})}{m_d(\mathbf{N})-m_d(\mathbf{w})-\alpha} \binom{m_d(\mathbf{v})}{\alpha}}.
    \end{align*}
    Notice each factor corresponds to a Krawtchouk polynomial of the form $K_{m_d(\mathbf{N})-m_d(\mathbf{w})}\prt{m_d(\mathbf{v});m_d(\mathbf{N})}$. Therefore, the shadow coefficients can be rewritten as
    \begin{equation}\label{eq:S=Adashes}
    \begin{split}
    S_\mathbf{w}(M,N)&=\sum_{S\subseteq [n]}\corch{\prod_{d\in \mathbb{D}} K_{m_d(\mathbf{N})-m_d(\mathbf{w})}\prt{m_d(\mathcal{D}(S));m_d(\mathbf{N})}}\ptr{S}{\ptr{S^c}{M}\ptr{S^c}{N}}\\
        &=\sum_{\mathbf{v}\subseteq\mathbf{N}}\sum_{\mathcal{D}(S)=\mathbf{v}}\corch{\prod_{d\in \mathbb{D}} K_{m_d(\mathbf{N})-m_d(\mathbf{w})}\prt{m_d(\mathbf{v});m_d(\mathbf{N})}}\ptr{S}{\ptr{S^c}{M}\ptr{S^c}{N}}\\
        &=\sum_{\mathbf{v}\subseteq\mathbf{N}}\corch{\prod_{d\in \mathbb{D}} K_{m_d(\mathbf{N})-m_d(\mathbf{w})}\prt{m_d(\mathbf{v});m_d(\mathbf{N})}}\sum_{\mathcal{D}(S)=\mathbf{v}}\ptr{S}{\ptr{S^c}{M}\ptr{S^c}{N}}\\
        &=\sum_{\mathbf{v}\subseteq\mathbf{N}}\corch{\prod_{d\in \mathbb{D}} K_{m_d(\mathbf{N})-m_d(\mathbf{w})}\prt{m_d(\mathbf{v});m_d(\mathbf{N})}}A'_\mathbf{v}(M,N).
    \end{split}
    \end{equation}
    
    This means we can also express the shadow weight enumerator in terms of the unitary one. By substituting the relation from Eq. \eqref{eq:S=Adashes} into the definition of the multivariate shadow weight enumerator, we obtain
\begin{align*}
    S_{MN}\prt{\vec{x},\vec{y}}&=\sum_{\mathbf{w}\subseteq \mathbf{N}}S_{\mathbf{w}}(M,N)\prod_{d\in \mathbb{D}} x_d^{m_d(\mathbf{N})-m_d(\mathbf{w})}y_d^{m_d(\mathbf{w})}\\
    &=\sum_{\mathbf{w}\subseteq \mathbf{N}}\claud{\sum_{\mathbf{v}\subseteq\mathbf{N}}\corch{\prod_{d\in \mathbb{D}} K_{m_d(\mathbf{N})-m_d(\mathbf{w})}\prt{m_d(\mathbf{v});m_d(\mathbf{N})}}A'_\mathbf{v}(M,N)}\prod_{d\in \mathbb{D}} x_d^{m_d(\mathbf{N})-m_d(\mathbf{w})}y_d^{m_d(\mathbf{w})}.
\end{align*}

Next, we exchange the order of summation between $\mathbf{w}$ and $\mathbf{v}$. By pulling the unitary coefficient $A'_\mathbf{v}(M,N)$ out of the inner sum and grouping the dimension-specific terms together into a single product, this yields
\begin{equation*}
    S_{MN}\prt{\vec{x},\vec{y}} = \sum_{\mathbf{v}\subseteq \mathbf{N}} A'_\mathbf{v}(M,N) \sum_{\mathbf{w}\subseteq \mathbf{N}}\corch{\prod_{d\in\mathbb{D}}K_{m_d(\mathbf{N})-m_d(\mathbf{w})}\prt{m_d(\mathbf{v});m_d(\mathbf{N})} x_d^{m_d(\mathbf{N})-m_d(\mathbf{w})}y_d^{m_d(\mathbf{w})}}.
\end{equation*}

Because the summation over the sub-multisets $\mathbf{w} \subseteq \mathbf{N}$ decouples across each distinct dimension, it is equivalent to summing independently over the integer multiplicities $m_d(\mathbf{w})$ from $0$ up to $m_d(\mathbf{N})$. Exchanging the product and the sum gives
\begin{equation*}
    S_{MN}\prt{\vec{x},\vec{y}} = \sum_{\mathbf{v}\subseteq \mathbf{N}} A'_\mathbf{v}(M,N) \prod_{d\in\mathbb{D}}\corch{ \sum_{m_d(\mathbf{w})=0}^{m_d(\mathbf{N})}K_{m_d(\mathbf{N})-m_d(\mathbf{w})}\prt{m_d(\mathbf{v});m_d(\mathbf{N})} x_d^{m_d(\mathbf{N})-m_d(\mathbf{w})}y_d^{m_d(\mathbf{w})}}.
\end{equation*}

The inner sum now exactly matches the generating function property of the Krawtchouk polynomials. Evaluating this generating function maps the sum into a product of binomials
\begin{equation*}
    S_{MN}\prt{\vec{x},\vec{y}} = \sum_{\mathbf{v}\subseteq \mathbf{N}} A'_\mathbf{v}(M,N) \prod_{d\in\mathbb{D}} (x_d + y_d)^{m_d(\mathbf{N})-m_d(\mathbf{v})} (y_d - x_d)^{m_d(\mathbf{v})}.
\end{equation*}

Finally, recognizing this form as the original multivariate unitary weight enumerator evaluated at a transformed set of variables, we arrive at the identity
\begin{equation}\label{eq:SMN=AdashMN}
    S_{MN}\prt{\claud{x_d}_{d\in\mathbb{D}},\claud{y_d}_{d\in\mathbb{D}}} = A'_{MN}\prt{\claud{x_d+y_d}_{d\in\mathbb{D}},\claud{y_d-x_d}_{d\in\mathbb{D}}}.
\end{equation}
    Combining these results yields
    \begin{align*}
    S_{MN}\prt{\claud{x_d}_{d\in\mathbb{D}}, \claud{y_d}_{d\in\mathbb{D}}} &= A'_{MN}\prt{\claud{x_d+y_d}_{d\in\mathbb{D}}, \claud{y_d-x_d}_{d\in\mathbb{D}}} \\
    &= A_{MN}\prt{\claud{\frac{(d-1)x_d+(d+1)y_d}{d}}_{d\in\mathbb{D}}, \claud{\frac{y_d-x_d}{d}}_{d\in\mathbb{D}}},
\end{align*}
    where we have used Eq. \eqref{eq:SMN=AdashMN} and Eq. \eqref{eq:A=Adash_polynomial}.
\end{proof}
\subsection{Linear transformations}
Notice that Eq. \eqref{eq:AdashW=sumAV} reflects how the unitary coefficients $A'_\mathbf{w}$ can be expressed as a linear combination of Shor-Laflamme multiset coefficients $A_\mathbf{v}$. A direct application of the mixed-dimensional quantum MacWilliams identity in Theorem \ref{thm:mixed-dimensionalMacWilliams} allows us invert this relation to compute Shor-Laflamme coefficients given the unitary ones.
\begin{lemma}\label{lemma:inversion}
For any dimension multiset $\mathbf{w} \subseteq \mathbf{N}$, the Shor-Laflamme multiset coefficients $A_{\mathbf{w}}$ can be explicitly computed from the unitary multiset coefficients $A'_{\mathbf{v}}$ as
\begin{equation*}
    A_{\mathbf{w}} (M,N)= \sum_{\mathbf{v} \subseteq \mathbf{w}} (-1)^{|\mathbf{w}| - |\mathbf{v}|} \corch{\prod_{d \in \mathbb{D}} \binom{m_d(\mathbf{N}) - m_d(\mathbf{v})}{m_d(\mathbf{w}) - m_d(\mathbf{v})} d^{m_d(\mathbf{v})}} A'_{\mathbf{v}}(M,N).
\end{equation*}
\end{lemma}

\begin{proof}
From Eq. \eqref{eq:A=Adash_polynomial}, expanding both sides using the definitions of the multivariate polynomials we have
\begin{equation*}
    \sum_{\mathbf{v}\subseteq \mathbf{N}}A'_\mathbf{v}(M,N)\prod_{d\in\mathbb{D}}\prt{x_d-y_d}^{m_d(\mathbf{N})-m_d(\mathbf{v})}\prt{dy_d}^{m_d(\mathbf{v})}=\sum_{\mathbf{w}\subseteq \mathbf{N}}A_\mathbf{w}(M,N)\prod_{d\in  \mathbb{D}}x_d^{m_d(\mathbf{N})-m_d(\mathbf{w})}y_d^{m_d(\mathbf{w})}.
\end{equation*}
To find $A_{\mathbf{w}}(M,N)$, we identify the coefficient of the monomial $\prod_{d\in\mathbb{D}}x_d^{m_d(\mathbf{N})-m_d(\mathbf{w})}y_d^{m_d(\mathbf{w})}$ on the left-hand side. Expanding the term $(x_d - y_d)^{m_d(\mathbf{N}) - m_d(\mathbf{v})}$ using the binomial theorem we obtain
\begin{align*}
    \text{LHS} &= \sum_{\mathbf{v}\subseteq \mathbf{N}}A'_\mathbf{v}(M,N)\prod_{d\in\mathbb{D}} \corch{ d^{m_d(\mathbf{v})}y_d^{m_d(\mathbf{v})} \sum_{\alpha=0}^{m_d(\mathbf{N})-m_d(\mathbf{v})} \binom{m_d(\mathbf{N})-m_d(\mathbf{v})}{\alpha} x_d^{m_d(\mathbf{N})-m_d(\mathbf{v})-\alpha} (-y_d)^\alpha } \\
    &= \sum_{\mathbf{v}\subseteq \mathbf{N}}A'_\mathbf{v}(M,N)\prod_{d\in\mathbb{D}} \corch{ \sum_{\alpha=0}^{m_d(\mathbf{N})-m_d(\mathbf{v})} \binom{m_d(\mathbf{N})-m_d(\mathbf{v})}{\alpha} (-1)^\alpha d^{m_d(\mathbf{v})} x_d^{m_d(\mathbf{N})-(m_d(\mathbf{v})+\alpha)} y_d^{m_d(\mathbf{v})+\alpha} }.
\end{align*}
We now equate powers of $y_d$ by setting $m_d(\mathbf{w}) = m_d(\mathbf{v}) + \alpha$, which implies $\alpha = m_d(\mathbf{w}) - m_d(\mathbf{v})$. Since $\alpha \ge 0$, this requires $m_d(\mathbf{w}) \ge m_d(\mathbf{v})$, corresponding to the multiset inclusion $\mathbf{v} \subseteq \mathbf{w}$. Substituting $\alpha$ we extract the coefficient
\begin{equation*}
    A_{\mathbf{w}}(M,N) = \sum_{\mathbf{v} \subseteq \mathbf{w}} A'_{\mathbf{v}}(M,N) \prod_{d\in\mathbb{D}} \binom{m_d(\mathbf{N}) - m_d(\mathbf{v})}{m_d(\mathbf{w}) - m_d(\mathbf{v})} (-1)^{m_d(\mathbf{w}) - m_d(\mathbf{v})} d^{m_d(\mathbf{v})}.
\end{equation*}
Noticing that $\sum_{d \in \mathbb{D}} (m_d(\mathbf{w}) - m_d(\mathbf{v})) = |\mathbf{w}| - |\mathbf{v}|$ and factoring the global sign $(-1)^{|\mathbf{w}| - |\mathbf{v}|}$ out of the product completes the proof.
\end{proof}
Building on these results, we can now use the MacWilliams identity from Theorem \ref{thm:mixed-dimensionalMacWilliams} and shadow identity from Theorem \ref{thm:shadow_identity} to express both the multiset dual and shadow coefficients explicitly as linear combinations of the primal Shor-Laflamme multiset coefficients via generalized Krawtchouk polynomials.
\begin{proposition}\label{prop:Bs_Ss_from_AsKrawtchouk}
    Given $\mathbf{v}\subseteq \mathbf{N}$ the following relations for the dual and shadow multiset coefficients hold
    \begin{align}
        B_{\mathbf{v}}(M,N) &= \prt{\dim [n]}^{-1} \sum_{\mathbf{w}\subseteq \mathbf{N}} A_\mathbf{w} (M,N) \prod_{d\in\mathbb{D}} \tilde{K}_{m_d(\mathbf{v})}\prt{m_d(\mathbf{w});m_d(\mathbf{N}),1,d^2-1}\label{eq:B=A_linear},\\
        S_{\mathbf{v}}(M,N) &= \prt{\dim [n]}^{-1} \sum_{\mathbf{w}\subseteq \mathbf{N}} (-1)^{|\mathbf{w}|}A_\mathbf{w}(M,N) \prod_{d\in\mathbb{D}} \tilde{K}_{m_d(\mathbf{v})}\prt{m_d(\mathbf{w});m_d(\mathbf{N}),d-1,d+1}\label{eq:S=A_linear},
    \end{align}
    where $\tilde{K}_j\prt{k;n,\gamma,\delta}$ are the generalized Krawtchouk polynomials defined as 
    \begin{equation*}
\tilde{K}_j\prt{k;n,\gamma,\delta}=\sum_{\alpha}\prt{-1}^\alpha \binom{n-k}{j-\alpha}\binom{k}{\alpha}\gamma^{\corch{(n-k)-(j-\alpha)}}\delta^{(j-\alpha)}.
    \end{equation*}
\end{proposition}
\begin{proof}
We start by proving Eq. \eqref{eq:B=A_linear}. Notice that, from the symmetry of the quantum MacWilliams identity proof we can invert the relation and state that 
\begin{equation*}
B_{MN}\prt{\claud{x_d}_{d\in\mathbb{D}},\claud{y_d}_{d\in\mathbb{D}}}=A_{MN}\prt{\claud{\frac{x_d+(d^2-1)y_d}{d}}_{d\in\mathbb{D}},\claud{\frac{x_d-y_d}{d}}_{d\in\mathbb{D}}}.
\end{equation*}
We will drop the $(M,N)$ notation for practical reasons. Expanding both sides according to Definition \ref{def:ShorLaflamme_unitary_weight_enumerators} and developing we obtain
\begin{align*}
\sum_{\mathbf{v} \subseteq \mathbf{N}}B_\mathbf{v}\prod_{d\in\mathbb{D}}x_d^{m_d(\mathbf{N})-m_d(\mathbf{v})}y_d^{m_d(\mathbf{v})}&=\sum_{\mathbf{w} \subseteq \mathbf{N}}A_\mathbf{w}\prod_{d\in\mathbb{D}}\prt{\frac{x_d+(d^2-1)y_d}{d}}^{m_d(\mathbf{N})-m_d(\mathbf{w})}\prt{\frac{x_d-y_d}{d}}^{m_d(\mathbf{w})}\\
&=\sum_{\mathbf{w} \subseteq \mathbf{N}}A_\mathbf{w}\prod_{d\in\mathbb{D}}\frac{1}{d^{m_d(\mathbf{N})}}\prt{x_d+(d^2-1)y_d}^{m_d(\mathbf{N})-m_d(\mathbf{w})}\prt{x_d-y_d}^{m_d(\mathbf{w})}\\
&=\frac{1}{\dim[n]}\sum_{\mathbf{w} \subseteq \mathbf{N}}A_\mathbf{w}\prod_{d\in\mathbb{D}}\corch{\sum_{j_d=0}^{m_d(\mathbf{N})} \tilde{K}_{j_d}\prt{m_d(\mathbf{w});m_d(\mathbf{N}),1,d^2-1}x_d^{m_d(\mathbf{N})-j_d}y_d^{j_d}},
\end{align*}
where on the last equality we have used the property $\sum_{j}\tilde{K}_j\prt{k;n,\gamma,\delta}x^{n-j}y^j=\prt{\gamma x+\delta y}^{n-k}\prt{x-y}^k$ derived in \cite{HuberEltschkaSiewertGuhne2018}. 

The expression can be further simplified by expanding the product of sums. This yields a multiple sum over all possible combinations of indices $j_d$ for each $d \in \mathbb{D}$, i.e.,
\begin{align*}
\text{RHS}&=\prt{\dim [n]}^{-1}\sum_{\mathbf{w} \subseteq \mathbf{N}}A_\mathbf{w} \corch{ \sum_{j_{d_1}=0}^{m_{d_1}(\mathbf{N})} \dots \sum_{j_{d_{|\mathbb{D}|}}=0}^{m_{d_{|\mathbb{D}|}}(\mathbf{N})} \prod_{d\in\mathbb{D}} \tilde{K}_{j_d}\prt{m_d(\mathbf{w});m_d(\mathbf{N}),1,d^2-1}x_d^{m_d(\mathbf{N})-j_d}y_d^{j_d} }.
\end{align*}
Notice that summing over all valid tuples of multiplicities $(j_{d_1}, \dots, j_{d_{|\mathbb{D}|}})$ is equivalent to summing over all possible sub-multisets $\mathbf{u} \subseteq \mathbf{N}$, where we identify $j_d = m_d(\mathbf{u})$ for all $d \in \mathbb{D}$. Applying this substitution and swapping the order of summation yields
\begin{align*}
\text{RHS}&=\prt{\dim [n]}^{-1}\sum_{\mathbf{w} \subseteq \mathbf{N}}A_\mathbf{w}\sum_{\mathbf{u}\subseteq \mathbf{N}}\corch{\prod_{d\in\mathbb{D}} \tilde{K}_{m_d(\mathbf{u})}\prt{m_d(\mathbf{w});m_d(\mathbf{N}),1,d^2-1}x_d^{m_d(\mathbf{N})-m_d(\mathbf{u})}y_d^{m_d(\mathbf{u})}}\\
&=\prt{\dim [n]}^{-1}\sum_{\mathbf{u\subseteq N}}\corch{\sum_{\mathbf{w}\subseteq \mathbf{N}}A_\mathbf{w}\prod_{d\in\mathbb{D}} \tilde{K}_{m_d(\mathbf{u})}\prt{m_d(\mathbf{w});m_d(\mathbf{N}),1,d^2-1}}\prod_{d\in \mathbb{D}}x_d^{m_d(\mathbf{N})-m_d(\mathbf{u})}y_d^{m_d(\mathbf{u})}.
\end{align*}
By equating the coefficients of the monomials $\prod_{d} x_d^{m_d(\mathbf{N})-m_d(\mathbf{v})} y_d^{m_d(\mathbf{v})}$ on both sides, we obtain the explicit linear relation
\begin{equation*}
B_{\mathbf{v}} = \prt{\dim [n]}^{-1} \sum_{\mathbf{w}\subseteq \mathbf{N}} A_\mathbf{w} \prod_{d\in\mathbb{D}} \tilde{K}_{m_d(\mathbf{v})}\prt{m_d(\mathbf{w});m_d(\mathbf{N}),1,d^2-1}.
\end{equation*}
Now we prove Eq. \eqref{eq:S=A_linear} by proceeding analogously:
\begin{align*}
\sum_{\mathbf{v} \subseteq \mathbf{N}}S_\mathbf{v}\prod_{d\in\mathbb{D}}x_d^{m_d(\mathbf{N})-m_d(\mathbf{v})}y_d^{m_d(\mathbf{v})}&=\sum_{\mathbf{w} \subseteq \mathbf{N}}A_\mathbf{w}\prod_{d\in\mathbb{D}}\prt{\frac{(d-1)x_d+(d+1)y_d}{d}}^{m_d(\mathbf{N})-m_d(\mathbf{w})}\prt{\frac{y_d-x_d}{d}}^{m_d(\mathbf{w})}\\
&=\sum_{\mathbf{w} \subseteq \mathbf{N}}A_\mathbf{w}\prod_{d\in\mathbb{D}}\frac{(-1)^{m_d(\mathbf{w})}}{d^{m_d(\mathbf{N})}}\prt{(d-1)x_d+(d+1)y_d}^{m_d(\mathbf{N})-m_d(\mathbf{w})}\prt{x_d-y_d}^{m_d(\mathbf{w})},
\end{align*}
where we have factored out a $(-1)$ to match the polynomial identity form $(x_d - y_d)$. Swapping the sum and the product and taking into account that $\prod_{d\in\mathbb{D}}(-1)^{m_d(\mathbf{w})}=(-1)^{\sum_{d\in\mathbb{D}} m_d(\mathbf{w})}=(-1)^{|\mathbf{w}|}$ we obtain
\begin{align*}
\text{RHS}&=\prt{\dim [n]}^{-1}\sum_{\mathbf{w} \subseteq \mathbf{N}}(-1)^{|\mathbf{w}|}A_\mathbf{w}\prod_{d\in\mathbb{D}}\corch{\sum_{j_d=0}^{m_d(\mathbf{N})} \tilde{K}_{j_d}\prt{m_d(\mathbf{w});m_d(\mathbf{N}),d-1,d+1}x_d^{m_d(\mathbf{N})-j_d}y_d^{j_d}}\\
&=\prt{\dim [n]}^{-1}\sum_{\mathbf{u\subseteq N}}\corch{\sum_{\mathbf{w}\subseteq \mathbf{N}}A_\mathbf{w}(-1)^{|\mathbf{w}|}\prod_{d\in\mathbb{D}} \tilde{K}_{m_d(\mathbf{u})}\prt{m_d(\mathbf{w});m_d(\mathbf{N}),d-1,d+1}}\prod_{d\in \mathbb{D}}x_d^{m_d(\mathbf{N})-m_d(\mathbf{u})}y_d^{m_d(\mathbf{u})}.
\end{align*}
As before, equating the coefficients proves the claim.
\end{proof}
\section{Bounds on mixed-dimensional codes}\label{sec:bounds_on_mixed-dimensional_codes}
The enumerator identities derived in the previous section provide the algebraic foundation for determining the physical limits of quantum codes. In this section, we derive generalized versions of the linear program in \cite{HuberEltschkaSiewertGuhne2018} and the Hamming and Singleton bounds, showing how the multiset approach accounts for the non-uniform resource costs of heterogeneous systems.
\subsection{Linear program}

Following the methodology established in \cite{HuberEltschkaSiewertGuhne2018}, we can formulate a linear program to systematically rule out the existence of quantum codes with specific parameters. To be physically valid, the Shor-Laflamme multiset coefficients of any hypothetical code must satisfy the general positivity and minimum distance constraints outlined in Proposition \ref{prop:coeff_properties}, alongside the structural symmetries imposed by the mixed-dimensional MacWilliams identity (Theorem \ref{thm:mixed-dimensionalMacWilliams}) and the mixed-dimensional shadow identity (Theorem \ref{thm:shadow_identity}). By leveraging Proposition \ref{prop:Bs_Ss_from_AsKrawtchouk}, these polynomial identities can be cast directly as linear constraints on the primal coefficients $A_{\mathbf{w}}(\Pi_\mathcal{Q})$. Combining these elements yields the central feasibility theorem of this section.

\begin{theorem}\label{thm:LP}
    Given the parameters $((D_1,D_2,\dots, D_n), K, D)$ of a hypothetical quantum error-correcting code, if there is no set of valid coefficients $\{A_{\mathbf{v}}\}_{\mathbf{v}\subseteq \mathbf{N}}$ such that 
    \begin{equation}
    \begin{aligned}
        A_{\emptyset} &= KB_{\emptyset} = K^2, \\
        A_\mathbf{v} &= KB_{\mathbf{v}} \ge 0, \quad &&\forall \mathbf{v}\subseteq \mathbf{N}: \quad 1 < \dim{\mathbf{v}} < D,\\
        KB_{\mathbf{v}} &\ge A_\mathbf{v} \ge 0, \quad &&\forall \mathbf{v}\subseteq \mathbf{N}: \quad \dim{\mathbf{v}} \ge D,\\
        S_\mathbf{v} &\ge 0, \quad &&\forall \mathbf{v}\subseteq \mathbf{N},
    \end{aligned}
    \end{equation}
    where the dual coefficients $B_{\mathbf{v}}$ and shadow  coefficients $S_{\mathbf{v}}$ are strictly defined by the linear transformations
    \begin{equation*}
    \begin{split}
        B_{\mathbf{v}} &= \prt{\dim [n]}^{-1} \sum_{\mathbf{w}\subseteq \mathbf{N}} A_\mathbf{w} \prod_{d\in\mathbb{D}} \tilde{K}_{m_d(\mathbf{v})}\prt{m_d(\mathbf{w});m_d(\mathbf{N}),1,d^2-1},\\
        S_{\mathbf{v}} &= \prt{\dim [n]}^{-1} \sum_{\mathbf{w}\subseteq \mathbf{N}} (-1)^{|\mathbf{w}|}A_\mathbf{w} \prod_{d\in\mathbb{D}} \tilde{K}_{m_d(\mathbf{v})}\prt{m_d(\mathbf{w});m_d(\mathbf{N}),d-1,d+1},
    \end{split}
    \end{equation*}
    then the code does not exist.
\end{theorem}
If we restrict our search specifically to pure quantum error-correcting codes, the conditions become even tighter. In a pure code, the scalar product of any non-trivial detectable error with the code projector strictly vanishes (see Eq. \eqref{eq:knill-laflamme_conditions}). This immediately forces the corresponding multiset coefficients to zero, leading to the following corollary.

\begin{corollary}\label{cor:pure_LP}
    If the hypothetical quantum error-correcting code in Theorem \ref{thm:LP} is assumed to be pure, the linear program is further constrained by the strict condition:
    \begin{equation}
        A_\mathbf{v} = KB_{\mathbf{v}} = 0, \quad \forall \mathbf{v}\subseteq \mathbf{N}: \quad 1 < \dim{\mathbf{v}} < D.
    \end{equation}
    If no valid set of coefficients $\{A_{\mathbf{v}}\}_{\mathbf{v}\subseteq \mathbf{N}}$ satisfies the general constraints of Theorem \ref{thm:LP} alongside this purity constraint, then a pure code with parameters $((D_1,D_2,\dots, D_n), K, D)$ does not exist.
\end{corollary}
\subsection{Quantum Hamming bound}
Our multiset approach allows us to prove not only weight-related bounds but also geometric volume bounds, such as the quantum Hamming bound for mixed dimensions. To do so, we must distinguish between pure and non-degenerate codes. While a pure code requires the code projector to yield zero information about any detectable error algebraically meaning $c_E=0$ in the Knill-Laflamme conditions from Eq. \eqref{eq:knill-laflamme_conditions}, a code is said to be \textit{non-degenerate} with respect to a set of correctable errors if every error in that set maps the code subspace $\mathcal{Q}$ to a distinct, mutually orthogonal subspace within the total Hilbert space \cite{Gottesman2026}. This geometric property allows us to sum the dimensions of the error spaces to establish a strict upper bound. 

We adapt the proof from \cite{EkertMacchiavello96} to our mixed-dimensional requirements.
\begin{theorem}
    Let $\mathbb{H} = \bigotimes_{i=1}^n \mathbb{C}^{D_i}$ be a mixed-dimensional Hilbert space characterized by the total dimension multiset $\mathbf{N}$ and a set of distinct dimensions $\mathbb{D}$. Suppose a non-degenerate quantum error-correcting code encodes a $K$-dimensional subspace $\mathcal{Q} \subseteq \mathbb{H}$ with a dimensional minimum distance $D$. If the code is designed to correct all errors up to a dimensional weight threshold $T$, where $T^2 < D$, then the code dimension is bounded by
    \begin{equation*}
        K \sum_{\substack{\mathbf{v} \subseteq \mathbf{N} \\ \dim \mathbf{v} \le T}} \prod_{d \in \mathbb{D}} \binom{m_d(\mathbf{N})}{m_d(\mathbf{v})} (d^2 - 1)^{m_d(\mathbf{v})} \le \prod_{d \in \mathbb{D}} d^{m_d(\mathbf{N})}.
    \end{equation*}
\end{theorem}
 
\begin{proof}
    Let the quantum error-correcting code encode a $K$-dimensional subspace. For the code to successfully correct a set of errors in a non-degenerate manner, each correctable error must map the code subspace $\mathcal{Q}$ to a distinct, mutually orthogonal $K$-dimensional subspace within the total Hilbert space. 

    To establish the bound, we must count the total number of correctable errors. By assumption, the code corrects all errors whose dimension multiset $\mathbf{v}$ satisfies $\dim \mathbf{v} \le T$. For a fixed dimension multiset $\mathbf{v} \subseteq \mathbf{N}$, the number of valid subsystem supports is determined by $C_{\emptyset,\mathbf{v}}=\prod_{d\in \mathbb{D}}\binom{m_d(\mathbf{N})}{m_d(\mathbf{v})}$. 

    On each individual qudit of dimension $d$ within that chosen support, there are exactly $(d^2 - 1)$ non-trivial, linearly independent generalized Pauli errors. Therefore, the total number of distinct errors acting exactly on a subsystem with multiset $\mathbf{v}$ is given by
    \begin{equation*}
        \prod_{d \in \mathbb{D}} \binom{m_d(\mathbf{N})}{m_d(\mathbf{v})} (d^2 - 1)^{m_d(\mathbf{v})}.
    \end{equation*}

    The total number of correctable errors is the sum of these possibilities over all valid error profiles $\mathbf{v} \subseteq \mathbf{N}$ such that $\dim \mathbf{v} \le T$. Note that the ``no error'' case (the identity operator) corresponds to $\mathbf{v} = \emptyset$, which has dimension $\dim \emptyset = 1 \le T$, naturally evaluating to $1$ in our product.

    Because the code is non-degenerate, every single one of these errors requires its own dedicated $K$-dimensional orthogonal subspace. The sum of the dimensions of all these mutually orthogonal error subspaces cannot exceed the total available dimension of the Hilbert space, $\dim[n]$. Noticing that $\dim[n] = \prod_{d \in \mathbb{D}} d^{m_d(\mathbf{N})}$, we arrive at the inequality
    \begin{equation*}
        K \sum_{\substack{\mathbf{v} \subseteq \mathbf{N} \\ \dim \mathbf{v} \le T}} \prod_{d \in \mathbb{D}} \binom{m_d(\mathbf{N})}{m_d(\mathbf{v})} (d^2 - 1)^{m_d(\mathbf{v})} \le \prod_{d \in \mathbb{D}} d^{m_d(\mathbf{N})},
    \end{equation*}
    which completes the proof.
\end{proof}

\begin{example}\label{example:hamming22d}
    Let us evaluate the system composed of two qubits and one qudit of dimension $d \ge 5$, characterized by the total dimension multiset $\mathbf{N} = \{2, 2, d\}$ and a total physical dimension of $\dim[n] = 4d$. 
    
    Suppose we design a non-degenerate code with a dimensional minimum distance of $D=5$. For a set of errors to be strictly correctable, their maximum dimensional weight $T$ must satisfy $T^2 < D$. Since $T^2 < 5$ and $T$ must be an integer, the maximum correctable threshold is $T = 2$. Physically, this means the code can successfully correct a single qubit error, but it cannot reliably correct an error on the qudit.

    To find the maximum encodable dimension $K$, we evaluate the sum of the error volumes for all sub-multisets $\mathbf{v} \subseteq \mathbf{N}$ such that $\dim \mathbf{v} \le 2$:
    \begin{itemize}
        \item No errors ($\mathbf{v} = \emptyset$): The dimension is $\dim \emptyset = 1$. The volume is $\binom{2}{0}(2^2 - 1)^0 \times \binom{1}{0}(d^2 - 1)^0 = 1 \times 1 = 1$.
        \item One qubit error ($\mathbf{v} = \{2\}$): The dimension is $\dim \{2\} = 2$. The volume is $\binom{2}{1}(2^2 - 1)^1 \times \binom{1}{0}(d^2 - 1)^0 = 2 \times 3 \times 1 = 6.$
    \end{itemize}
    Any other error profile, such as an error on the qudit ($\mathbf{v} = \{d\}$) or on both qubits ($\mathbf{v} = \{2,2\}$), has a dimension of $4$ or higher and strictly exceeds the correctable threshold $T=2$.

    Summing these volumes, the total number of required orthogonal $K$-dimensional subspaces is $1 + 6 = 7$. Applying the Hamming bound inequality yields
    \begin{equation*}
        K \le \frac{4d}{7}.
    \end{equation*}
    Because the code dimension $K$ must be an integer, we conclude that $K \le \lfloor \frac{4d}{7} \rfloor$. 
\end{example}
\subsection{Quantum Singleton bound} 
This new approach allows us to prove other bounds such as the quantum Singleton bound for mixed-dimensional systems proven in \cite{BallZhang2026}. This proof is inspired by the Singleton bound proof given in \cite{HuberGrassl2020} for homogeneous systems.
\begin{theorem}\label{thm:singleton_bound}
If there exists a quantum error-correcting code defined over a system with total dimension multiset $\mathbf{N}$, code dimension $K$ and dimensional minimum distance $D$, then for any partition of the total multiset into three disjoint sub-multisets $\mathbf{N} = \mathbf{w}_1 \sqcup \mathbf{w}_2 \sqcup \mathbf{w}_3$ satisfying

\begin{equation*}
    \dim \mathbf{w}_1 < D \quad \text{and} \quad \dim \mathbf{w}_2 < D,
\end{equation*}
it must hold that
\begin{equation*}
    \dim \mathbf{w}_3 \ge K.
\end{equation*}
\end{theorem}

\begin{proof}
Let $\Pi_\mathcal{Q}$ be the projector onto a QECC $\mathcal{Q}$. Let us define the normalized unitary coefficient $\bar{A}'_{\mathbf{w}}(\Pi_\mathcal{Q})$ by dividing the standard unitary multiset coefficient by the total number of subsystems $C_{\emptyset, \mathbf{w}}$
\begin{equation*}
    \bar{A}'_{\mathbf{w}}= \frac{A'_{\mathbf{w}}}{C_{\emptyset, \mathbf{w}}} = (\dim \mathbf{w})^{-1} \sum_{\mathbf{v} \subseteq \mathbf{w}} \frac{C_{\mathbf{v}, \mathbf{w}}}{C_{\emptyset, \mathbf{w}}} A_{\mathbf{v}},
\end{equation*}
where we have dropped the $\Pi_\mathcal{Q}$ notation.
From Theorem \ref{thm:LP}, we know that $A_{\mathbf{v}} \le K B_{\mathbf{v}}$ holds for all $\mathbf{v} \subseteq \mathbf{N}$. Furthermore, by definition, $B'_{\mathbf{w}} = A'_{\mathbf{N}\setminus\mathbf{w}}$. By the symmetry of the  binomial coefficients, $C_{\emptyset, \mathbf{w}} = C_{\emptyset, \mathbf{N}\setminus\mathbf{w}}$. Dividing the universal bound $A'_{\mathbf{w}} \le K B'_{\mathbf{w}}$ by this constant shows that the normalized coefficients satisfy the relation
\begin{equation}\label{eq:singleton_normalized_bound}
    \bar{A}'_{\mathbf{w}} \le K \bar{A}'_{\mathbf{N}\setminus\mathbf{w}},
\end{equation}
with equality holding if $\dim \mathbf{w} < D$. 

Consider two multisets $\mathbf{w}$ and $\mathbf{u}$. By analyzing the ratio of the combinatorial factors inside the summation, we have
\begin{align*}
    \frac{C_{\mathbf{v}, \mathbf{w}}}{C_{\emptyset, \mathbf{w}}} &= \prod_{d\in \mathbb{D}} \corch{ \frac{m_d(\mathbf{w})! \prt{m_d(\mathbf{N})-m_d(\mathbf{v})}!}{\prt{m_d(\mathbf{w})-m_d(\mathbf{v})}! m_d(\mathbf{N})!} } \\
    &= \prod_{d\in \mathbb{D}} \corch{ m_d(\mathbf{w})\prt{m_d(\mathbf{w})-1}\cdots \prt{m_d(\mathbf{w})-m_d(\mathbf{v})+1}\frac{\prt{m_d(\mathbf{N})-m_d(\mathbf{v})}!}{m_d(\mathbf{N})!} }\\
    &\le \prod_{d\in \mathbb{D}} \corch{ \prt{m_d(\mathbf{w})+m_d(\mathbf{u})}\cdots \prt{m_d(\mathbf{w})+m_d(\mathbf{u})-m_d(\mathbf{v})+1}\frac{\prt{m_d(\mathbf{N})-m_d(\mathbf{v})}!}{m_d(\mathbf{N})!} }= \frac{C_{\mathbf{v}, \mathbf{w}\sqcup \mathbf{u}}}{C_{\emptyset, \mathbf{w}\sqcup \mathbf{u}}}.
\end{align*}

Because $A_{\mathbf{v}} \ge 0$ for all $\mathbf{v} \subseteq \mathbf{N}$, restricting the summation index from $\mathbf{w} \sqcup \mathbf{u}$ to $\mathbf{w}$ strictly bounds the sum from below. Therefore,
\begin{equation}\label{eq:singleton_monotonicity}
    \begin{split}
    \bar{A}'_{\mathbf{w} \sqcup \mathbf{u}} &= \prt{\dim(\mathbf{w}\sqcup\mathbf{u})}^{-1} \sum_{\mathbf{v}\subseteq \mathbf{w}\sqcup\mathbf{u}} \frac{C_{\mathbf{v}, \mathbf{w}\sqcup \mathbf{u}}}{C_{\emptyset, \mathbf{w}\sqcup \mathbf{u}}}A_\mathbf{v}\ge \prt{\dim \mathbf{w}}^{-1} \prt{\dim \mathbf{u}}^{-1}  \sum_{\mathbf{v}\subseteq \mathbf{w}} \frac{C_{\mathbf{v}, \mathbf{w}\sqcup \mathbf{u}}}{C_{\emptyset, \mathbf{w}\sqcup \mathbf{u}}}A_\mathbf{v}  \\
    &\ge \prt{\dim \mathbf{w}}^{-1} \prt{\dim \mathbf{u}}^{-1}  \sum_{\mathbf{v}\subseteq \mathbf{w}} \frac{C_{\mathbf{v}, \mathbf{w}}}{C_{\emptyset, \mathbf{w}}}A_\mathbf{v} = (\dim \mathbf{u})^{-1} \bar{A}'_{\mathbf{w}}.
    \end{split}
\end{equation}

We can now prove the bound. Let us partition the total system multiset into three disjoint subsets $\mathbf{N} = \mathbf{w}_1 \sqcup \mathbf{w}_2 \sqcup \mathbf{w}_3$, with respective dimensions $d_1$, $d_2$ and $d_3$. We assume $\mathbf{w}_1$ and $\mathbf{w}_2$ are correctable, meaning $d_1 < D$ and $d_2 < D$. 

Applying the equality from Eq. \eqref{eq:singleton_normalized_bound} for $d_1 < D$ yields $\bar{A}'_{\mathbf{w}_1} = K \bar{A}'_{\mathbf{w}_2 \sqcup \mathbf{w}_3}$, and applying the monotonicity condition from Eq. \eqref{eq:singleton_monotonicity} to the right-hand side gives $\bar{A}'_{\mathbf{w}_1} = K \bar{A}'_{\mathbf{w}_2 \sqcup \mathbf{w}_3} \ge \frac{K}{d_3} \bar{A}'_{\mathbf{w}_2}$.

By applying the exact same logic using the equality for the second correctable subsystem $d_2 < D$, we obtain the symmetric bound $\bar{A}'_{\mathbf{w}_2} = K \bar{A}'_{\mathbf{w}_1 \sqcup \mathbf{w}_3} \ge \frac{K}{d_3} \bar{A}'_{\mathbf{w}_1}$.

Substituting $\bar{A}'_{\mathbf{w}_2}$ back into the first equation yields $\bar{A}'_{\mathbf{w}_1} \ge \frac{K}{d_3} \prt{\frac{K}{d_3} \bar{A}'_{\mathbf{w}_1}} = \frac{K^2}{d_3^2} \bar{A}'_{\mathbf{w}_1}$.

Because $\bar{A}'_{\mathbf{w}_1}$ contains the empty set term, it is strictly positive. We can therefore divide both sides by $\bar{A}'_{\mathbf{w}_1}$ to obtain
\begin{equation*}
    1 \ge \frac{K^2}{d_3^2}, 
\end{equation*}
which implies $d_3 \ge K$ and the proof is finished.
\end{proof}
Interestingly, for pure codes, the same kind of proof allows us to obtain a tighter bound, as the next corollary suggests.
\begin{corollary}\label{cor:pure_singleton}
If the quantum error-correcting code is pure, then for any sub-multiset $\mathbf{s} \subset \mathbf{N}$ satisfying $\dim \mathbf{s} < D$, the code dimension is bounded by
\begin{equation*}
    K \le \frac{\dim \mathbf{N}}{(\dim \mathbf{s})^2}.
\end{equation*}
\end{corollary}

\begin{proof}
Let $\mathbf{s} \subset \mathbf{N}$ such that $\dim \mathbf{s} < D$. Because the code is pure, $A_{\mathbf{v}} = 0$ for all proper sub-multisets $\emptyset \neq \mathbf{v} \subseteq \mathbf{s}$. Consequently, the normalized coefficient for $\mathbf{s}$ collapses entirely to the empty-set term
\begin{equation*}
    \bar{A}'_{\mathbf{s}} = (\dim \mathbf{s})^{-1} \sum_{\mathbf{v} \subseteq \mathbf{s}} \frac{C_{\mathbf{v}, \mathbf{s}}}{C_{\emptyset, \mathbf{s}}} A_{\mathbf{v}} = (\dim \mathbf{s})^{-1} A_{\emptyset} = \frac{K^2}{\dim \mathbf{s}}.
\end{equation*}

Applying the strict equality from Eq. \eqref{eq:singleton_normalized_bound} for $\dim \mathbf{s} < D$ yields $\bar{A}'_{\mathbf{s}} = K \bar{A}'_{\mathbf{N}\setminus\mathbf{s}}$. Substituting and solving for $\bar{A}'_{\mathbf{N}\setminus\mathbf{s}}$ yields
\begin{equation*}
 \bar{A}'_{\mathbf{N}\setminus\mathbf{s}} = \frac{K}{\dim \mathbf{s}}.
\end{equation*}

Now, consider the expansion of the normalized coefficient for the complement $\mathbf{N} \setminus \mathbf{s}$. Since $A_{\mathbf{v}} \ge 0$ for all multisets, we can lower-bound this sum by discarding all terms except the empty-set term $A_\emptyset$ which is equal to $K^2$. That is
\begin{equation*}
    \bar{A}'_{\mathbf{N}\setminus\mathbf{s}} = \prt{\dim (\mathbf{N}\setminus\mathbf{s})}^{-1} \sum_{\mathbf{v} \subseteq \mathbf{N}\setminus\mathbf{s}} \frac{C_{\mathbf{v}, \mathbf{N}\setminus\mathbf{s}}}{C_{\emptyset, \mathbf{N}\setminus\mathbf{s}}} A_{\mathbf{v}} \ge \prt{\dim (\mathbf{N}\setminus\mathbf{s})}^{-1} A_{\emptyset} = \frac{K^2}{\dim (\mathbf{N}\setminus\mathbf{s})}.
\end{equation*}

Equating this lower bound with the exact value we found above gives
\begin{equation*}
    \frac{K}{\dim \mathbf{s}} \ge \frac{K^2}{\dim (\mathbf{N}\setminus\mathbf{s})}.
\end{equation*}

Using the total dimension relation $\dim (\mathbf{N}\setminus\mathbf{s}) = \frac{\dim \mathbf{N}}{\dim \mathbf{s}}$ and rearranging proves the claim.
\end{proof}
\begin{example}
    Consider again a hypothetical mixed-dimensional quantum error-correcting code defined over two qubits and a qudit of dimension $d \ge 5$, giving a total dimension multiset $\mathbf{N} = \claud{2, 2, d}$. The total dimension of the system is $\dim \mathbf{N} = 2^2 \cdot d = 4d$. Suppose we require the code to have a dimensional distance threshold of $D = 5$. We can analyze the constraints on the code dimension $K$ using both the general and the pure mixed-dimensional quantum Singleton bounds.

    First, let us assume the code is pure. To obtain the tightest bound, we must identify the sub-multiset $\mathbf{s} \subset \mathbf{N}$ with the largest physical dimension strictly less than $D=5$. Since $d \ge 5$, the optimal choice is the subset of the two qubits, $\mathbf{s} = \claud{2, 2}$, which yields $\dim \mathbf{s} = 4 < 5$. Applying the mixed-dimensional pure Singleton bound gives
    \begin{equation*}
        K \le \frac{\dim \mathbf{N}}{(\dim \mathbf{s})^2} = \frac{4d}{4^2} = \frac{d}{4}.
    \end{equation*}
    Since the code dimension must be an integer, any pure code with these parameters is strictly restricted to $K \le \lfloor \frac{d}{4} \rfloor$. 

    Next, we evaluate the general mixed-dimensional quantum Singleton bound without the purity assumption. We must partition the total multiset into three disjoint sub-multisets $\mathbf{N} = \mathbf{w}_1 \sqcup \mathbf{w}_2 \sqcup \mathbf{w}_3$ such that $\dim \mathbf{w}_1 < 5$ and $\dim \mathbf{w}_2 < 5$. An optimal partition is $\mathbf{w}_1 = \claud{2}$,  $\mathbf{w}_2 = \claud{2}$ and $\mathbf{w}_3 = \claud{d}$. The general bound dictates that $K \le \dim \mathbf{w}_3$, which restricts the code dimension to $$K \le d.$$

    This discrepancy highlights a profound feature of mixed-dimensional spaces. The general Singleton bound theoretically leaves mathematical room for a code up to $K=d$. However, as evaluated previously for this exact system in Example \ref{example:hamming22d}, the geometric volume-packing constraint given by the quantum Hamming bound strictly restricts non-degenerate codes to $K \le \lfloor \frac{4d}{7} \rfloor$. Because $\lfloor \frac{4d}{7} \rfloor < d$ for all $d \ge 5$, the Hamming bound rules out a large family of parameters that the general Singleton bound theoretically permits (for a non-degenerate code). 
    
    Furthermore, because $\lfloor \frac{d}{4} \rfloor < \lfloor \frac{4d}{7} \rfloor$ for all $d \ge 5$, any valid non-degenerate code that successfully maximizes its dimension up to the Hamming bound is mathematically forced to be impure. For instance, if $d=5$, the general Singleton bound permits $K \le 5$, but the Hamming bound strictly limits non-degenerate codes to $K \le 2$. Since purity implies $K \le 1$, a non-degenerate code achieving $K=2$ is mathematically forced to be an impure code. None of these codes are ruled out by the Linear Program in Theorem \ref{thm:LP}, except the pure code for $K=2$. This makes total sense because purity implies non-degeneracy, but there exist non-degenerate impure codes \cite{Cao2022}. 
\end{example}

From the previous example an important observation follows.
\begin{remark}
In homogeneous Hilbert spaces, saturation of the quantum Singleton bound strictly requires the code to be pure \cite{Rains99,HuberGrassl2020}, meaning quantum maximum distance separable codes (QMDS) are pure. Corollary \ref{cor:pure_singleton} establishes a profound departure from this rule in mixed-dimensional systems: if the maximal pure-code Singleton bound $K_{\text{pure}}$ is strictly less than the general bound $K_{\text{general}}$ from Theorem \ref{thm:singleton_bound}, any mixed-dimensional code saturating the Singleton bound is mathematically forced to be impure.
\end{remark}

\section{Absolutely maximally entangled states}\label{sec:ame_states}
While the previous sections focused on the constraints governing quantum error-correcting codes, we now turn our attention to absolutely maximally entangled (AME) states. In this section, we utilize the multiset weight enumerator machinery to establish generalized Scott and shadow bounds for mixed-dimensional entanglement and introduce a combinatorial grid method for the explicit construction of tripartite AME states.

Recall that we say a state $\ket{\psi}\in \mathbb{H}$ is \textit{absolutely maximally entangled} (AME) \cite{HuberEltschkaSiewertGuhne2018, BallZhang2026} if 
\begin{equation*}
\rho_S=\frac{1}{\dim_S}\mathds{1}_S
\end{equation*}
for all subsystems $S\subseteq [n]$ for which $\dim S\le \Delta \coloneqq \sqrt{\dim[n]}$. If this is the case, notice the purity is easily computed as $\tra{\rho_S^2}=(\dim S)^{-1}$.

If $\ket{\psi}$ is an AME state, notice that for subsystems for which $\dim S >\Delta$, its complement satisfies $\dim S^c \le \Delta$, which means $\rho_{S^c}=\frac{1}{\dim S^c}\mathds{1}_{S^c}$. From the Schmidt decomposition it can be seen that the purities hold the relation $\tra{\rho_S^2}=\tra{\rho_{S^c}^2}=(\dim S^c)^{-1}$.

This implies that 
\begin{equation*}
\mathcal{A}'_S(\ket{\psi})=\tra{\rho_S^2}=\frac{1}{\min\prt{\dim S, \dim S^c}}
\end{equation*}
for all $S\subseteq [n]$. So the unitary coefficients for AME states can be computed as

\begin{equation}\label{eq:Adash_AME}
    \begin{split}
        A'_\mathbf{w}(\ket{\psi})&=\sum_{\mathcal{D}(S)=\mathbf{w}} \mathcal{A}'_S(\ket{\psi})=\sum_{\mathcal{D}(S)=\mathbf{w}} \frac{1}{\min\prt{\dim S, \dim S^c}}=\frac{1}{\min\prt{\prod_{d\in \mathbf{w}}d, \frac{\dim[n]}{\prod_{d\in \mathbf{w}}d}}}\sum_{\mathcal{D}(S)=\mathbf{w}}1,\\
        &=\frac{1}{\min\prt{\dim\mathbf{w}, \frac{\dim[n]}{\dim\mathbf{w}}}} \prod_{d\in\mathbb{D}} \binom{m_d(\mathbf{N})}{m_d(\mathbf{w})},
    \end{split}
\end{equation}
where on the second line we have added the combinatorial factor $C_{\emptyset,\mathbf{w}}$ counting the number of subsystems with fixed dimensional multiset $\mathbf{w}$.

Furthermore, for an AME state $\ket{\psi}$, any reduction to a subsystem $S$ with $\dim S \le \Delta=\sqrt{\dim [n]}$ is maximally mixed, i.e., $\rho_S = (\dim S)^{-1}\mathds{1}_S$. Since the non-identity operators in our local basis $\mathcal{E}$ are traceless, $\tra{E\rho} = \tra{E_S \rho_S} = 0$ for all $E \neq \mathds{1}$ supported on $S$. Consequently, the coefficients vanish for these subsystems: $A_{\mathbf{w}}(\ket{\psi}) = 0$ for all multisets $\mathbf{v}$ such that $0 < \dim \mathbf{w} \le \Delta$. This follows also from Eq. \eqref{eq:Adash_AME}.

Notice that Eq. \eqref{eq:Adash_AME} provides an explicit expression for the unitary coefficients. Applying now Lemma \ref{lemma:inversion} we are able to find an explicit expression for the Shor-Laflamme coefficients for AME states. This is
\begin{equation*}
    A_{\mathbf{w}}(\ket{\psi}) = \sum_{\mathbf{v} \subseteq \mathbf{w}} \frac{(-1)^{|\mathbf{w}| - |\mathbf{v}|}}{\min\prt{\prod_{d\in \mathbf{v}}d, \frac{\dim[n]}{\prod_{d\in \mathbf{v}}d}}} \corch{\prod_{d \in \mathbb{D}} \binom{m_d(\mathbf{N}) - m_d(\mathbf{v})}{m_d(\mathbf{w}) - m_d(\mathbf{v})}\binom{m_d(\mathbf{N})}{m_d(\mathbf{v})} d^{m_d(\mathbf{v})}},
\end{equation*}
which can be simplified using 
\begin{equation*}
\binom{m_d(\mathbf{N}) - m_d(\mathbf{v})}{m_d(\mathbf{w}) - m_d(\mathbf{v})}\binom{m_d(\mathbf{N})}{m_d(\mathbf{v})}=\binom{m_d(\mathbf{N})}{m_d(\mathbf{w})}\binom{m_d(\mathbf{w})}{m_d(\mathbf{v})}
\end{equation*} 
to 
\begin{equation*}
    A_{\mathbf{w}}(\ket{\psi}) = \prt{\prod_{d \in \mathbb{D}} \binom{m_d(\mathbf{N})}{m_d(\mathbf{w})}} \sum_{\mathbf{v} \subseteq \mathbf{w}} \frac{(-1)^{|\mathbf{w}| - |\mathbf{v}|}}{\min\prt{\dim \mathbf{v}, \frac{\dim[n]}{\dim \mathbf{v}}}} \prod_{d \in \mathbb{D}} \binom{m_d(\mathbf{w})}{m_d(\mathbf{v})} d^{m_d(\mathbf{v})},
\end{equation*}
where $\dim \mathbf{v} = \prod_{d \in \mathbb{D}} d^{m_d(\mathbf{v})}$. 

This formulation is particularly instructive. Choose $\emptyset\subset\mathbf{w}\subseteq \mathbf{N}$ such that $\dim \mathbf{w}\le \sqrt{\dim [n]}$. Hence, since $\mathbf{v}\subseteq\mathbf{w}$, we have that $\dim \mathbf{v}\le\sqrt{\dim [n]}$. Therefore the $\min$ function in the equation above always evaluates to $\dim \mathbf{v}$. Note that $\dim \mathbf{v}=\prod_{d\in \mathbb{D}}d^{m_d(\mathbf{v})}$, which means we can further simplify to
\begin{align*}
A_\mathbf{w}\prt{\ket{\psi}}= \prt{\prod_{d \in \mathbb{D}} \binom{m_d(\mathbf{N})}{m_d(\mathbf{w})}} \sum_{\mathbf{v} \subseteq \mathbf{w}} (-1)^{|\mathbf{w}| - |\mathbf{v}|} \prod_{d \in \mathbb{D}} \binom{m_d(\mathbf{w})}{m_d(\mathbf{v})}.
\end{align*}
Using that $\sum_{d \in \mathbb{D}} (m_d(\mathbf{w}) - m_d(\mathbf{v})) = |\mathbf{w}| - |\mathbf{v}|$ and swapping the product and the sum we obtain
\begin{align*}
A_\mathbf{w}\prt{\ket{\psi}}&= \prt{\prod_{d \in \mathbb{D}} \binom{m_d(\mathbf{N})}{m_d(\mathbf{w})}} \sum_{\mathbf{v} \subseteq \mathbf{w}} \prod_{d \in \mathbb{D}} (-1)^{m_d(\mathbf{w})-m_d(\mathbf{v})}\binom{m_d(\mathbf{w})}{m_d(\mathbf{v})}\\&= \prt{\prod_{d \in \mathbb{D}} \binom{m_d(\mathbf{N})}{m_d(\mathbf{w})}} \prod_{d \in \mathbb{D}}\corch{\sum_{m_d(\mathbf{v})=0}^{m_d(\mathbf{w})}  (-1)^{m_d(\mathbf{w})-m_d(\mathbf{v})}\binom{m_d(\mathbf{w})}{m_d(\mathbf{v})}},
\end{align*}
which corresponds to the binomial expression of $(x+y)^n$ for $x=1$ and $y=-1$. Therefore $$A_\mathbf{w}\prt{\ket{\psi}}= 0,$$ for all $\mathbf{w}$ such that $\dim \mathbf{w}\le\Delta= \sqrt{\dim [n]}$, as mentioned above.
\subsection{Scott bound}

The enumerator framework developed in the previous sections allows us to derive strict algebraic inequalities that must be satisfied for any mixed-dimensional AME state to exist. This serves as a direct generalization of the bounds found by Scott for homogeneous systems \cite{Scott2004}.

\begin{theorem}\label{thm:mixed_scott_bound}
    Let $\mathbb{H} = \bigotimes_{i=1}^n \mathbb{C}^{D_i}$ be a mixed-dimensional Hilbert space. If an AME state $\ket{\psi} \in \mathbb{H}$ exists, then for any two nested dimension multisets $\mathbf{w}_1 \subset \mathbf{w}_2 \subseteq \mathbf{N}$ satisfying:
    \begin{enumerate}
        \item $\dim \mathbf{w}_1 > \Delta$,
        \item $\dim \mathbf{v} \le \Delta$ for all proper sub-multisets $\mathbf{v} \subset \mathbf{w}_1$,
        \item $\dim \mathbf{u} \le \Delta$ for all proper sub-multisets $\mathbf{u} \subset \mathbf{w}_2$ with $\mathbf{u} \neq \mathbf{w}_1$,
    \end{enumerate}
    the dimensions must satisfy the inequality
    \begin{equation}\label{eq:mixed_scott_inequality}
        \frac{(\dim \mathbf{w}_2)^2}{\dim[n]} - 1 \ge \prod_{d \in \mathbb{D}} \binom{m_d(\mathbf{w}_2)}{m_d(\mathbf{w}_1)} \prt{ \frac{(\dim \mathbf{w}_1)^2}{\dim[n]} - 1 }.
    \end{equation}
\end{theorem}
\begin{proof}
Consider $\mathbf{w}_1\subseteq \mathbf{N}$ such that $\dim \mathbf{w}_1> \Delta$ and $\dim \mathbf{v}\le \Delta$ for every proper sub-multiset $\mathbf{v}\subset \mathbf{w}_1$. From the discussion in Eq. \eqref{eq:Adash_AME},
\begin{equation}\label{eq:AdashW1}
    A'_{\mathbf{w}_1}=\prod_{d\in \mathbb{D}}\binom{m_d(\mathbf{N})}{m_d(\mathbf{w}_1)} \frac{\dim \mathbf{w}_1}{\dim [n]}.
\end{equation}
Using Eq. \eqref{eq:AdashW=sumAV}, we can expand $A'_{\mathbf{w}_1}$ in terms of the Shor-Laflamme operators $A_\mathbf{v}$ with $\mathbf{v}\subseteq \mathbf{w}_1$.  However, the only ones surviving are $A_\emptyset$ and $A_{\mathbf{w}_1}$ itself since the rest satisfy $0<\dim \mathbf{v}\le \Delta$. Therefore
\begin{equation*}
 A'_{\mathbf{w}_1}=(\dim \mathbf{w}_1)^{-1}\prt{C_{\emptyset, \mathbf{w}_1}A_\emptyset+C_{\mathbf{w}_1,\mathbf{w}_1}A_{\mathbf{w}_1}}=(\dim \mathbf{w}_1)^{-1}\prt{\prod_{d\in \mathbb{D}}\binom{m_d(\mathbf{N})}{m_d(\mathbf{w}_1)}A_\emptyset + A_{\mathbf{w}_1}}.
\end{equation*}
Substituting Eq. \eqref{eq:AdashW1} and solving for $A_{\mathbf{w}_1}$, we obtain
\begin{equation}\label{eq:Aw1_explicit}
A_{\mathbf{w}_1} = \prod_{d\in \mathbb{D}}\binom{m_d(\mathbf{N})}{m_d(\mathbf{w}_1)} \prt{\frac{(\dim \mathbf{w}_1)^2}{\dim [n]} - 1}.
\end{equation}
Now consider $\mathbf{w}_2\supset \mathbf{w}_1$ such that $\dim \mathbf{u}\le \Delta$ for all $\mathbf{u}\subseteq \mathbf{w}_2$ such that $\mathbf{u}\neq\mathbf{w}_1$. Now the only non vanishing terms in the expansion of $A'_{\mathbf{w}_2}$ are $A_\emptyset$, $A_{\mathbf{w}_1}$ and $A_{\mathbf{w}_2}$. That is
\begin{equation*}
A'_{\mathbf{w}_2}=\prt{\dim \mathbf{w}_2}^{-1}\prt{\prod_{d\in \mathbb{D}}\binom{m_d(\mathbf{N})}{m_d(\mathbf{w}_2)}+C_{\mathbf{w}_1,\mathbf{w}_2}A_{\mathbf{w}_1}+A_{\mathbf{w}_2}}.
\end{equation*}
On the other hand, $A'_{\mathbf{w}_2}=\prod_{d\in \mathbb{D}}\binom{m_d(\mathbf{N})}{m_d(\mathbf{w}_2)}\frac{\dim \mathbf{w}_2}{\dim [n]}$, so solving now for $A_{\mathbf{w}_2}$ yields
   \begin{equation}\label{eq:Aw2_explicit}
        A_{\mathbf{w}_2} = \prod_{d\in \mathbb{D}}\binom{m_d(\mathbf{N})}{m_d(\mathbf{w}_2)} \prt{\frac{(\dim \mathbf{w}_2)^2}{\dim [n]} - 1} - C_{\mathbf{w}_1,\mathbf{w}_2} A_{\mathbf{w}_1}
    \end{equation}

Now, substituting $A_{\mathbf{w}_1}$ from Eq. \eqref{eq:Aw1_explicit} into Eq. \eqref{eq:Aw2_explicit} we obtain
\begin{equation*}
        A_{\mathbf{w}_2} = \prod_{d\in \mathbb{D}}\binom{m_d(\mathbf{N})}{m_d(\mathbf{w}_2)} \prt{\frac{(\dim \mathbf{w}_2)^2}{\dim [n]} - 1} - C_{\mathbf{w}_1,\mathbf{w}_2} \prod_{d\in \mathbb{D}}\binom{m_d(\mathbf{N})}{m_d(\mathbf{w}_1)} \prt{\frac{(\dim \mathbf{w}_1)^2}{\dim [n]} - 1}.
\end{equation*}
Notice that
\begin{align*}
        C_{\mathbf{w}_1,\mathbf{w}_2} \prod_{d\in \mathbb{D}}\binom{m_d(\mathbf{N})}{m_d(\mathbf{w}_1)} &= \prod_{d\in \mathbb{D}} \corch{\binom{m_d(\mathbf{N}) - m_d(\mathbf{w}_1)}{m_d(\mathbf{w}_2) - m_d(\mathbf{w}_1)} \binom{m_d(\mathbf{N})}{m_d(\mathbf{w}_1)}}= \prod_{d\in \mathbb{D}} \corch{\binom{m_d(\mathbf{N})}{m_d(\mathbf{w}_2)} \binom{m_d(\mathbf{w}_2)}{m_d(\mathbf{w}_1)}},
    \end{align*}
So
\begin{equation*}
        A_{\mathbf{w}_2} = \prod_{d\in \mathbb{D}}\binom{m_d(\mathbf{N})}{m_d(\mathbf{w}_2)} \corch{ \prt{\frac{(\dim \mathbf{w}_2)^2}{\dim [n]} - 1} - \prod_{d\in \mathbb{D}}\binom{m_d(\mathbf{w}_2)}{m_d(\mathbf{w}_1)} \prt{\frac{(\dim \mathbf{w}_1)^2}{\dim [n]} - 1} },
    \end{equation*}
and imposing non-negativity we obtain the bound
\begin{equation*}
\frac{(\dim \mathbf{w}_2)^2}{\dim [n]} - 1\ge \prod_{d\in \mathbb{D}}\binom{m_d(\mathbf{w}_2)}{m_d(\mathbf{w}_1)} \prt{\frac{(\dim \mathbf{w}_1)^2}{\dim [n]} - 1}.
\end{equation*}
\end{proof}
Now Scott bound from \cite{Scott2004} can be recovered easily.
\begin{corollary}
If $\ket{\psi}\in \bigotimes_{i=1}^n \mathbb{C}^q$ is an AME state, then 
\begin{align*}
    n\le \begin{cases*}
    2(q^2-1)&$n$ even,\\
    2q(q+1)-1&$n$ odd.\\
    \end{cases*}
\end{align*}
\end{corollary}
\begin{proof}
    Let $n$ be even, and choose standard weights $k_1 = \frac{n}{2} + 1$ and $k_2 = \frac{n}{2} + 2$. The multisets have dimensions $\dim \mathbf{w}_1 = q^{n/2 + 1}$ and $\dim \mathbf{w}_2 = q^{n/2 + 2}$. The total dimension is $\dim[n] = q^n$. The relative combinatorial factor evaluates to $\binom{n/2 + 2}{n/2 + 1} = \frac{n}{2} + 2$. Substituting into Eq. \eqref{eq:mixed_scott_inequality} yields
    \begin{align*}
        \frac{q^{n+4}}{q^n} - 1 &\ge \prt{\frac{n}{2} + 2} \prt{\frac{q^{n+2}}{q^n} - 1} \\
        q^4 - 1 &\ge \prt{\frac{n}{2} + 2} (q^2 - 1) \\
        (q^2 - 1)(q^2 + 1) &\ge \prt{\frac{n}{2} + 2} (q^2 - 1).
    \end{align*}
    Dividing by $(q^2 - 1)$ gives $q^2 + 1 \ge \frac{n}{2} + 2$, which simplifies to the exact bound $n \le 2(q^2 - 1)$. Applying the same logic for odd $n$ using $k_1 = \frac{n-1}{2} + 1$ yields $n \le 2q(q+1) - 1$.
\end{proof}
To understand how the mixed-dimensional multiset conditions naturally collapse to the standard choices in the homogeneous literature, consider a system where all $n$ parties share a constant local dimension $q$. In this limit, any dimension multiset $\mathbf{w}$ is entirely characterized by its cardinality $k = |\mathbf{w}|$, giving it a physical dimension of $\dim \mathbf{w} = q^k$. The AME threshold dimension is simply $\Delta = q^{n/2}$. Condition 1 requires $\dim \mathbf{w}_1 > \Delta$, which translates to $q^{k_1} > q^{n/2}$, or strictly $k_1 > n/2$. Simultaneously, Condition 2 demands that all proper sub-multisets, which contain at most $k_1 - 1$ qudits, must not exceed the threshold. This implies $q^{k_1 - 1} \le q^{n/2}$, which simplifies to $k_1 \le n/2 + 1$. The only integer $k_1$ that simultaneously satisfies both bounds, 
\begin{equation*}
    n/2 < k_1 \le n/2 + 1,
\end{equation*}
is exactly $k_1 = \lfloor n/2 \rfloor + 1$. This perfectly recovers the  weight selection used to derive the standard Scott bounds for both even and odd $n$. Furthermore, Condition 3 ensures that any larger target multiset $\mathbf{w}_2$ can only cross 
the threshold via $\mathbf{w}_1$, forcing its size to be exactly $k_2 = k_1 + 1$.
\begin{example}
    As an application of Theorem \ref{thm:mixed_scott_bound}, we can prove the non-existence of an AME state in a mixed-dimensional system composed of seven qubits and one qutrit. Let $\mathbb{H} = (\mathbb{C}^2)^{\otimes 7} \otimes \mathbb{C}^3$. The total dimension multiset is $\mathbf{N} = \claud{2, 2, 2, 2, 2, 2, 2, 3}$, giving a total system dimension of $\dim[n] = 2^7 \times 3 = 384$. The AME threshold is $\Delta = \sqrt{384} \approx 19.59$.
    
    To apply the bound, we select two nested multisets consisting entirely of qubits. Let $\mathbf{w}_1 = \claud{2, 2, 2, 2, 2}$ (five qubits) and $\mathbf{w}_2 = \claud{2, 2, 2, 2, 2, 2}$ (six qubits). We verify the three necessary conditions:
    \begin{enumerate}
        \item $\dim \mathbf{w}_1 = 2^5 = 32$. Since $32 > 19.59$, the first condition holds.
        \item The proper sub-multisets of $\mathbf{w}_1$ contain at most 4 qubits, meaning their maximum dimension is $2^4 = 16$. Since $16 \le 19.59$, the second condition is strictly met.
        \item The proper sub-multisets of $\mathbf{w}_2$, excluding $\mathbf{w}_1$, also contain at most 4 qubits, thus bounded by a dimension of 16. Since $16 \le 19.59$, the third condition is met.
    \end{enumerate}
    
    Since the conditions are satisfied, any hypothetical AME state must obey the inequality in Eq. \eqref{eq:mixed_scott_inequality}. The relative multiset binomial coefficient evaluates to choosing 5 qubits from a pool of 6, yielding
    \begin{equation*}
        \prod_{d\in \mathbb{D}}\binom{m_d(\mathbf{w}_2)}{m_d(\mathbf{w}_1)} = \binom{6}{5}\binom{0}{0} = 6.
    \end{equation*}
    Substituting the dimensions $\dim \mathbf{w}_2 = 64$, $\dim \mathbf{w}_1 = 32$ and $\dim[n] = 384$ into the bound we obtain
    \begin{align*}
        \frac{29}{3} &\ge 10. 
    \end{align*}
    This results in a mathematical contradiction. Therefore, an AME state for 7 qubits and 1 qutrit strictly does not exist.
\end{example}
\subsection{Shadow bounds}
In \cite{HuberEltschkaSiewertGuhne2018} it is proven that the shadow inequalities from Eq. \eqref{eq:generalized_shadow_inequality} provide tighter bounds than Scott's when considering AME states. Here we follow the same line of thought.

Using Eq. \eqref{eq:S=Adashes} and plugging Eq. \eqref{eq:Adash_AME} we can also relate the shadow coefficients and the unitary ones for AME states as
\begin{equation}\label{eq:Shadows_AME}
\begin{split}
    S_\mathbf{w}(\ket{\psi})=\sum_{\mathbf{v}\subseteq \mathbf{N}} \frac{1}{\min\prt{\dim\mathbf{v}, \frac{\dim[n]}{\dim\mathbf{v}}}} \prod_{d\in\mathbb{D}} \corch{K_{m_d(\mathbf{N})-m_d(\mathbf{w})}\prt{m_d(\mathbf{v});m_d(\mathbf{N})}\binom{m_d(\mathbf{N})}{m_d(\mathbf{v})}}.
\end{split}
\end{equation}
In particular, we can compute the first shadow coefficient for the empty set as
\begin{equation}\label{eq:shadow_empty_set_AME}
S_\emptyset(\ket{\psi})=\sum_{\mathbf{v}\subseteq \mathbf{N}} \frac{\prt{-1}^{|\mathbf{v}|}}{\min\prt{\dim\mathbf{v}, \frac{\dim[n]}{\dim\mathbf{v}}}} \prod_{d\in \mathbb{D}} \binom{m_d\prt{\mathbf{N}}}{m_d(\mathbf{v})}.
\end{equation}
For constant local dimensions, this expression reduces to Eq. (49) in \cite{HuberEltschkaSiewertGuhne2018}.
Note from Eq. \eqref{eq:Shadows_AME}, that $S_\mathbf{w}(\ket{\psi})$ solely depends on the multiplicities of the dimensions of the total Hilbert space and the particular multiset $\mathbf{w}$ for which the shadow coefficient is computed. Therefore, given the dimensions of a Hilbert, a negative shadow coefficient implies the non-existence of absolutely maximally entangled states in that Hilbert space.

An example is given in \cite{HuberEltschkaSiewertGuhne2018}. Here we give the explicit derivation in our multiset approach.
\begin{example} AME states in $\mathbb{H}=\mathbb{C}^2\otimes\mathbb{C}^2\otimes\mathbb{C}^2\otimes\mathbb{C}^3$ do not exist. We compute the shadow coefficient for the empty set according to Eq. \eqref{eq:shadow_empty_set_AME}, which yields 
\begin{align*}
    S_\emptyset&=\frac{(-1)^0}{\min\prt{1,\frac{24}{1}}}\binom{3}{0}\binom{1}{0}\\
    &\quad+\frac{(-1)^1}{\min\prt{2,\frac{24}{2}}}\binom{3}{1}\binom{1}{0}
    +\frac{{(-1)^1}}{\min\prt{3,\frac{24}{3}}}\binom{3}{0}\binom{1}{1}\\
    &\quad+\frac{(-1)^2}{\min\prt{2\cdot 3,\frac{24}{2\cdot 3}}}\binom{3}{1}\binom{1}{1}+\frac{(-1)^2}{\min\prt{2\cdot 2,\frac{24}{2\cdot 2}}}\binom{3}{2}\binom{1}{0}\\
    &\quad+\frac{(-1)^3}{\min\prt{2\cdot 2\cdot 2,\frac{24}{2\cdot 2\cdot 2}}}\binom{3}{3}\binom{1}{0}+\frac{(-1)^3}{\min\prt{2\cdot 2\cdot 3,\frac{24}{2\cdot 2\cdot 3}}}\binom{3}{2}\binom{1}{1}\\
    &\quad+\frac{(-1)^4}{\min\prt{2\cdot 2\cdot 2\cdot 3,\frac{24}{2\cdot 2\cdot 2\cdot 3}}}\binom{3}{3}\binom{1}{1}=1-\frac{3}{2}-\frac{1}{3}+\frac{3}{4}+\frac{3}{4}-\frac{1}{3}-\frac{3}{2}+1=-\frac{1}{6}\not\ge 0.
\end{align*}
Therefore AME states in $\mathbb{H}=\mathbb{C}^2\otimes\mathbb{C}^2\otimes\mathbb{C}^2\otimes\mathbb{C}^3$ do not exist.
\end{example}
This multivariate framework provides an efficient path to evaluating the full suite of shadow multiset enumerators for arbitrary party counts and local dimensions. The predictive power of this analytical approach is visualized in Fig. \ref{fig:ame_heatmap}, which maps the existence landscape of mixed-dimensional AME states formed by qubits and qutrits by identifying regions where the non-negativity of shadow enumerators is violated. For instance, no AME states exist for 8 parties consisting of qubits and qutrits.

A similar heatmap for qutrits and ququarts can be found in Appendix \ref{sec:qutrits_ququarts}.
\begin{figure}[ht]
    \centering
    \includegraphics[width=0.81\textwidth]{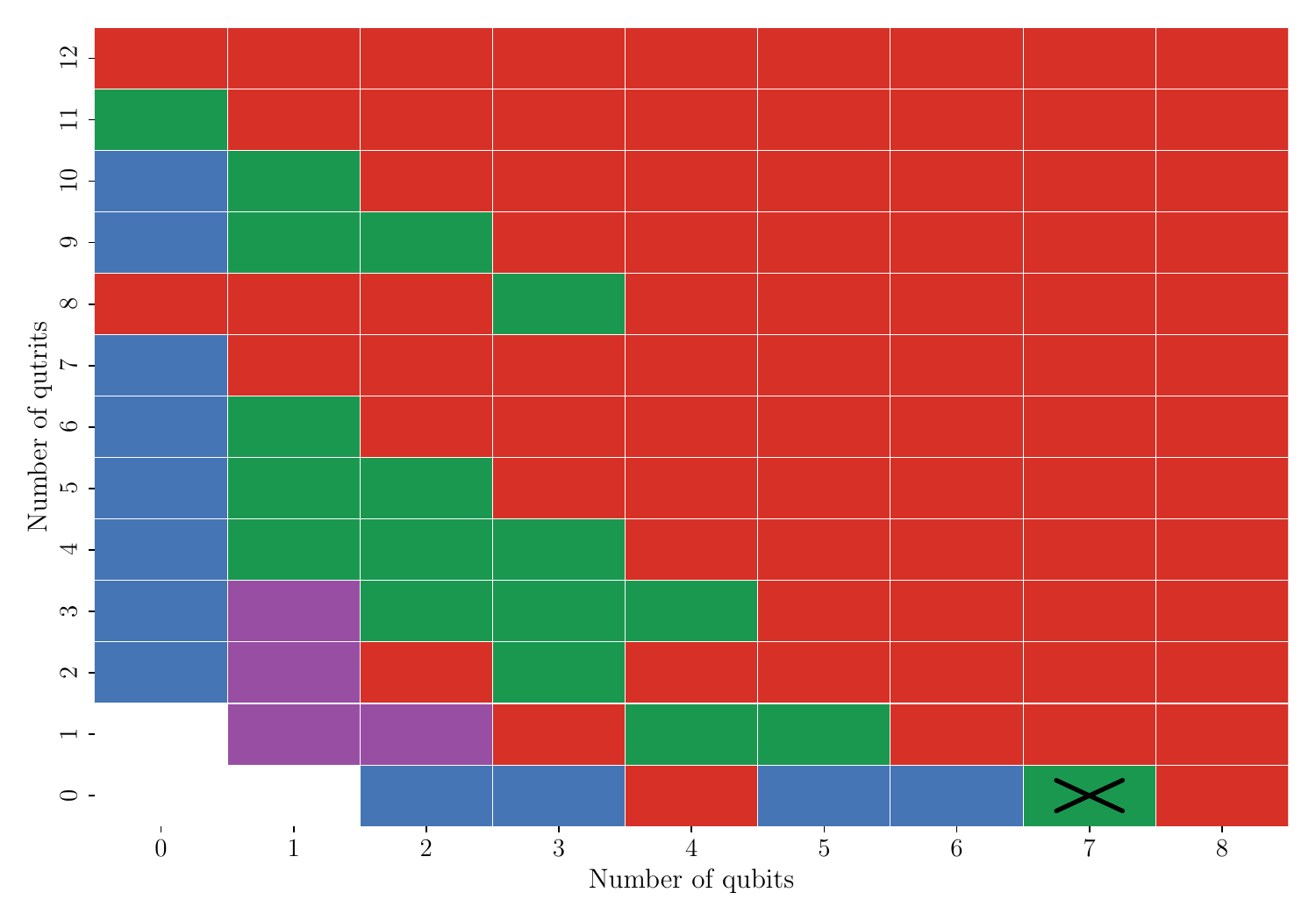}
    \caption{Existence of mixed-dimensional AME states formed by qubits and qutrits based on the positivity of shadow multiset enumerators. Green squares represent potentially existing states, while red squares indicate states strictly forbidden by negative shadow enumerators. Blue squares are already existing states in homogeneous systems \cite{TableAME, huber_quantum_code_bounds}. The seven-qubit AME state was ruled out by \cite{HuberOtfriedSiewert17}, marked with a cross. Purple squares represent existing AME states in heterogeneous scenarios. The $\mathbb{C}^2\otimes \mathbb{C}^3$ has the trivial AME $\frac{1}{\sqrt2}\prt{\ket{00}+\ket{11}}$, an embedded Bell state. An example for $\mathbb{C}^2\otimes \prt{\mathbb{C}^3}^{\otimes 3}$ is given in \cite{BallZhang2026} and \cite{HuberEltschkaSiewertGuhne2018}. Our construction of AME states in $\mathbb{C}^2\otimes \mathbb{C}^2\otimes \mathbb{C}^3$ and $\mathbb{C}^2\otimes \mathbb{C}^3\otimes \mathbb{C}^3$ can be found in Section \ref{sec:tripartiteAME} in Eq. \eqref{eq:AME223} and Eq. \eqref{eq:AME233}, respectively.}
    \label{fig:ame_heatmap}
\end{figure}

To end this section, let us recall our closed formulas for the mixed-dimensional coefficients for AME states $\ket{\psi}$ in the general setting $\mathbb{H}=\bigotimes_{i=1}^{n} \mathbb{C}^{D_i}$ for $\mathbf{w}\subseteq\mathbf{N}$ and $S\subseteq [n]$:
\begin{align}
A_{\mathbf{w}}(\ket{\psi}) &= \prt{\prod_{d \in \mathbb{D}} \binom{m_d(\mathbf{N})}{m_d(\mathbf{w})}} \sum_{\mathbf{v} \subseteq \mathbf{w}} \frac{(-1)^{|\mathbf{w}| - |\mathbf{v}|}}{\min\prt{\dim \mathbf{v}, \frac{\dim[n]}{\dim \mathbf{v}}}} \prod_{d \in \mathbb{D}} \binom{m_d(\mathbf{w})}{m_d(\mathbf{v})} d^{m_d(\mathbf{v})},\nonumber\\
A'_\mathbf{w}(\ket{\psi})&=\frac{1}{\min\prt{\dim \mathbf{w}, \frac{\dim[n]}{\dim \mathbf{w}}}} \prod_{d\in\mathbb{D}} \binom{m_d(\mathbf{N})}{m_d(\mathbf{w})},\nonumber\\
S_\mathbf{w}(\ket{\psi})&=\sum_{\mathbf{v}\subseteq \mathbf{N}}\frac{1}{\min\prt{\dim \mathbf{v}, \frac{\dim[n]}{\dim \mathbf{v}}}} \prod_{d\in\mathbb{D}}  \corch{K_{m_d(\mathbf{N})-m_d(\mathbf{w})}\prt{m_d(\mathbf{v});m_d(\mathbf{N})}\binom{m_d(\mathbf{N})}{m_d(\mathbf{v})}},\nonumber\\
\mathcal{A}'_S(\ket{\psi})&=\frac{1}{\min\prt{\dim S, \dim S^c}}.\label{eq:AME_cal_enumerators}
\end{align}
Once we obtain $A_\mathbf{w}$ for all sub-multisets $\mathbf{w}\subseteq\mathbf{N}$, we can recover the dual enumerators $B_\mathbf{w}$ using Proposition \ref{prop:Bs_Ss_from_AsKrawtchouk} and we obtain a complete characterization of the enumerators for AME states in any dimension. 

\begin{example}
Consider the AME state in $\mathbb{C}^2\otimes \mathbb{C}^3\otimes \mathbb{C}^3\otimes \mathbb{C}^3$ found in \cite{BallZhang2026}. That is
\begin{align*}
\ket{\psi}&=\frac{1}{\sqrt{12}}(\ket{0022}+\ket{0201}+\ket{0120}+\ket{0011}+\ket{0102}+\ket{0210}\\ 
&\qquad\qquad-\ket{1101}+\ket{1110}-\ket{1012}+\ket{1202}-\ket{1220}+\ket{1021}).
\end{align*}
Its corresponding multiset coefficients are given in Table \ref{tab:AME2333_enumerators}.
\begin{table}[ht]
    \centering
    \footnotesize 
    \setlength{\tabcolsep}{4pt} 
    \renewcommand{\arraystretch}{2}
    \caption{Complete weight enumerator profile for the mixed-dimensional AME state in $\mathbb{C}^2\otimes\mathbb{C}^3\otimes\mathbb{C}^3\otimes\mathbb{C}^3$ from \cite{BallZhang2026}. Shor-Laflamme ($A_{\mathbf{v}}$, $B_{\mathbf{v}}$), shadow ($S_{\mathbf{v}}$), unitary ($A'_{\mathbf{v}}$, $B'_{\mathbf{v}}$) and calligraphic ($\mathcal{A}'_S$, $\mathcal{B}'_S$) enumerators are evaluated. Subsystems $S \subseteq \{1, 2, 3, 4\}$ are grouped by their corresponding dimension multiset $\mathcal{D}(S)=\mathbf{v}$, where party 1 is the qubit and parties 2, 3, and 4 are the qutrits. Due to the symmetry of AME states, their calligraphic enumerators depend strictly on the dimension multiset of the subsystem as seen in Eq. \eqref{eq:AME_cal_enumerators}. Therefore, the unitary coefficient for a given multiset is simply the product of its corresponding calligraphic coefficients and the combinatorial number of subsystems sharing that multiset.}
    \label{tab:AME2333_enumerators} 
      \vspace{10pt}
    \begin{tabular}{ccccccccc}
        \toprule
        \textbf{Subsystems} $S$ & \textbf{Multiset} $\mathcal{D}(S)=\mathbf{v}$ & $A_{\mathbf{v}}$ & $B_{\mathbf{v}}$ & $S_{\mathbf{v}}$ & $A'_{\mathbf{v}}$ & $B'_{\mathbf{v}}$ & $\mathcal{A}'_S$ & $\mathcal{B}'_S$ \\
        \midrule
        $\emptyset$ & $\emptyset$ & $1$ & $1$ & $0$ & $1$ & $1$ & $1$ & $1$ \\
        $\{1\}$ & $\{2\}$ & $0$ & $0$ & $0$ & $\frac{1}{2}$ & $\frac{1}{2}$ & $\frac{1}{2}$ & $\frac{1}{2}$ \\
        $\{2\}$, $\{3\}$, $\{4\}$ & $\{3\}$ & $0$ & $0$ & $0$ & $1$ & $1$ & $\frac{1}{3}$ & $\frac{1}{3}$ \\
        $\{1, 2\}$, $\{1, 3\}$, $\{1, 4\}$ & $\{2, 3\}$ & $0$ & $0$ & $6$ & $\frac{1}{2}$ & $\frac{1}{2}$ & $\frac{1}{6}$ & $\frac{1}{6}$ \\
        $\{2, 3\}$, $\{2, 4\}$, $\{3, 4\}$ & $\{3, 3\}$ & $\frac{3}{2}$ & $\frac{3}{2}$ & $4$ & $\frac{1}{2}$ & $\frac{1}{2}$ & $\frac{1}{6}$ & $\frac{1}{6}$ \\
        $\{1, 2, 3\}$, $\{1, 2, 4\}$, $\{1, 3, 4\}$ & $\{2, 3, 3\}$ & $\frac{27}{2}$ & $\frac{27}{2}$ & $0$ & $1$ & $1$ & $\frac{1}{3}$ & $\frac{1}{3}$ \\
        $\{2, 3, 4\}$ & $\{3, 3, 3\}$ & $11$ & $11$ & $0$ & $\frac{1}{2}$ & $\frac{1}{2}$ & $\frac{1}{2}$ & $\frac{1}{2}$ \\
        $\{1, 2, 3, 4\}$ & $\{2, 3, 3, 3\}$ & $27$ & $27$ & $6$ & $1$ & $1$ & $1$ & $1$ \\
        \bottomrule
    \end{tabular}
\end{table}
\end{example}
\subsection{Tripartite AME states}\label{sec:tripartiteAME}
Let $\mathbb{H} = \mathbb{C}^{d_1}\otimes \mathbb{C}^{d_2}\otimes \mathbb{C}^{d_3}$ be our tripartite Hilbert space such that $d_1\le d_2 d_3$, $d_2\le d_1 d_3$ and $d_3 \le d_1d_2$. Without loss of generality we can assume $d_1\le d_2$ and our inequalities become
\begin{equation*}
\frac{d_1}{d_2}\le \frac{d_2}{d_1}\le d_3\le d_1 d_2.
\end{equation*}

Here, we formally define a constructive methodology to build absolutely maximally entangled (AME) states in this scenario by mapping the quantum algebraic constraints to a combinatorial grid problem. A similar construction is given in \cite{ShenChen21} with different constraints and not allowing instances with partially symmetric states, for example, states in which $d_1=d_2$.

\begin{theorem}\label{thm:grid_tripartiteAME}
An AME state in $\mathbb{C}^{d_1}\otimes \mathbb{C}^{d_2}\otimes \mathbb{C}^{d_3}$ with the dimensions satisfying $\frac{d_1}{d_2}\le \frac{d_2}{d_1}\le d_3\le d_1 d_2$ can be explicitly constructed if there exists a $d_1 \times d_2$ grid of non-negative probability weights $a_{jk} \ge 0$, structured as follows:

\begin{equation*}
\renewcommand{\arraystretch}{1.5} 
\begin{array}{c|ccccc|c}
    & \ket{0}_2 & \ket{1}_2 & \ket{2}_2 & \dots & \ket{d_2-1}_2 & \sum\\
    \hline
    \ket{0}_1 & a_{00} & a_{01}  & a_{02}  & \dots &  a_{0(d_2-1)} &\mathbf{1/d_1}\\
    \ket{1}_1 & a_{10} & a_{11}  & a_{12}  & \dots &  a_{1(d_2-1)}&\mathbf{1/d_1}\\
    \vdots & \vdots & \vdots & \vdots & \ddots & \vdots & \vdots \\
    \ket{d_1-1}_1 & a_{(d_1-1)0} & a_{(d_1-1)1}  & a_{(d_1-1)2}  & \dots &  a_{(d_1-1)(d_2-1)}&\mathbf{1/d_1}\\
    \hline 
    \sum& \mathbf{1/d_2} & \mathbf{1/d_2} & \mathbf{1/d_2} & \dots &\mathbf{1/d_2} &\mathbf{1}
\end{array}
\end{equation*}

The non-zero entries of this grid must be partitionable into $d_3$ disjoint sets, each labeled by a logical state $\ket{i}_3$ for $i=0,\dots,d_3-1$, simultaneously satisfying the following combinatorial rules:
\begin{enumerate}
    \item \textbf{Partition Normalization:} The sum of the weights inside each partition $\ket{i}_3$ is exactly $1/d_3$.
    \item \textbf{Row Balance:} The sum of the weights in any given row $j$ is exactly $1/d_1$.
    \item \textbf{Column Balance:} The sum of the weights in any given column $k$ is exactly $1/d_2$.
    \item \textbf{Strict Disjointedness:} No individual partition $\ket{i}_3$ contains more than one non-zero entry in any given row or column.
\end{enumerate}
\end{theorem}
Note that the combinatorial rules above provide sufficient conditions for constructing an AME state, but they are not strictly necessary. Specifically, Rule 4 (Strict Disjointedness) represents just one simple, systematic way to guarantee that the off-diagonal terms of the reduced density matrices vanish, though other cancellations using complex phases are possible.

\begin{proof}
We proceed constructively by assuming a $d_1 \times d_2$ grid of non-negative weights $a_{jk}$ exists and satisfies Rules 1–4. From this grid, we define $d_3$ disjoint partitions, each labeled by an index $i \in \{0, \dots, d_3-1\}$. For each partition $i$, we construct a $d_1 \times d_2$ matrix $\Phi_i$ whose elements are defined by the square roots of the grid weights, that is
\begin{equation*}
(\Phi_i)_{jk} = \begin{cases}
\sqrt{a_{jk}} & \text{if cell } (j,k) \text{ belongs to partition } i, \\
0 & \text{otherwise}.
\end{cases}
\end{equation*}
Using these matrices, we define the tripartite quantum state
\begin{equation*}
\ket{\psi} = \sum_{i=0}^{d_3-1} \sum_{j=0}^{d_1-1} \sum_{k=0}^{d_2-1} (\Phi_i)_{jk} \ket{j}_1\ket{k}_2\ket{i}_3.
\end{equation*}

To verify that $\ket{\psi}$ is an AME state, we must show that its reductions to any single subsystem $1, 2$ or $3$ are maximally mixed. Because the dimensions satisfy $d_1 \le d_2 \le d_3$ and $d_3 \le d_1 d_2$, all single-party reductions fall under the mixed-dimensional AME threshold, while all two-party reductions exceed it.

First, we evaluate the reduction to subsystem 3. We rewrite the global state as $\ket{\psi} = \frac{1}{\sqrt{d_3}}\sum_{i=0}^{d_3-1} \ket{\psi_i}\ket{i}_3$, where $\ket{\psi_i} = \sqrt{d_3} \sum_{j,k} (\Phi_i)_{jk} \ket{j}_1\ket{k}_2$. The inner product between any two such states is
\begin{equation*}
\braket{\psi_i}{\psi_m} = d_3 \sum_{j=0}^{d_1-1} \sum_{k=0}^{d_2-1} (\Phi_i)_{jk}^* (\Phi_m)_{jk}.
\end{equation*}
For $i \neq m$, Rule 4 dictates that the partitions are strictly disjoint, meaning $(\Phi_i)_{jk}$ and $(\Phi_m)_{jk}$ cannot be simultaneously non-zero for the same cell $(j,k)$. Thus, $\braket{\psi_i}{\psi_m} = 0$. For $i = m$, Rule 1 (Partition Normalization) dictates that $\sum_{j,k} |(\Phi_i)_{jk}|^2 = \sum_{(j,k) \in \text{partition } i} a_{jk} = \frac{1}{d_3}$. Therefore, $\braket{\psi_i}{\psi_i} = d_3 (\frac{1}{d_3}) = 1$. The states $\ket{\psi_i}$ form an orthonormal basis, confirming that $\rho_3 = \frac{1}{d_3}\mathds{1}_{d_3}$.

Second, tracing out subsystem 3 yields $\rho_{12} = \sum_{i=0}^{d_3-1} \sum_{j,j',k,k'} (\Phi_i)_{jk} (\Phi_i)_{j'k'}^* \ketbra{j}{j'}_1 \otimes \ketbra{k}{k'}_2$. Tracing out subsystem 2 gives the reduced state for subsystem 1 as
\begin{equation*}
\rho_1 = \sum_{j,j'=0}^{d_1-1} \left( \sum_{i=0}^{d_3-1} (\Phi_i \Phi_i^\dagger)_{jj'} \right) \ketbra{j}{j'}_1.
\end{equation*}
For the off-diagonal elements ($j \neq j'$), the coefficient expands to $\sum_i \sum_k (\Phi_i)_{jk} (\Phi_i)_{j'k}^*$. Rule 4 requires that no single partition $i$ contains more than one non-zero entry in any given column $k$. Thus, for any fixed $i$ and $k$, the product $(\Phi_i)_{jk} (\Phi_i)_{j'k}^*$ must be zero, meaning all off-diagonal terms strictly vanish. For the diagonal elements ($j = j'$), the coefficient evaluates to $\sum_i \sum_k |(\Phi_i)_{jk}|^2 = \sum_k a_{jk}$. By Rule 2 (Row Balance), this sum equals $1/d_1$. Therefore, $\rho_1 = \frac{1}{d_1}\mathds{1}_{d_1}$.

Finally, tracing out subsystem 1 from $\rho_{12}$ yields the reduced state for subsystem 2
\begin{equation*}
\rho_2 = \sum_{k,k'=0}^{d_2-1} \left( \sum_{i=0}^{d_3-1} (\Phi_i^\dagger \Phi_i)_{k'k} \right) \ketbra{k}{k'}_2.
\end{equation*}
For the off-diagonal elements ($k \neq k'$), the coefficient expands to $\sum_i \sum_j (\Phi_i)_{jk}^* (\Phi_i)_{jk'}$. Rule 4 requires that no partition $i$ contains more than one non-zero entry in any given row $j$, ensuring these cross-terms are identically zero. For the diagonal elements ($k = k'$), the coefficient is $\sum_i \sum_j |(\Phi_i)_{jk}|^2 = \sum_j a_{jk}$. By Rule 3 (Column Balance), this sum equals $1/d_2$. Therefore, $\rho_2 = \frac{1}{d_2}\mathds{1}_{d_2}$.

Since $\rho_1$, $\rho_2$, and $\rho_3$ are all maximally mixed, the constructed state $\ket{\psi}$ satisfies the condition for being an absolutely maximally entangled state, finishing the proof.
\end{proof}

To concretize this abstraction, let us apply the generalized combinatorial rules to explicitly construct an AME state.
\begin{example}
 Consider the Hilbert space $\mathbb{H} = \mathbb{C}^2 \otimes \mathbb{C}^3 \otimes \mathbb{C}^4$. In the language of our generalization, the local dimensions are $d_1 = 2$ (subsystem 1), $d_2 = 3$ (subsystem 2) and $d_3 = 4$ (subsystem 3). Following our established framework, we require a $d_1 \times d_2$ probability grid. Thus, we construct a single $2 \times 3$ grid, plotting the squared amplitudes. 

We must partition the populated cells into $d_3 = 4$ distinct sets. To do this, we append a subscript label $\ket{i}$ to each populated cell to denote its assignment to one of the four logical states of the third subsystem: $\ket{0}$, $\ket{1}$, $\ket{2}$ and $\ket{3}$.

The solution is given by the following table, where we have dropped the subindices indicating the subsystems.
\begin{equation*}
\renewcommand{\arraystretch}{1.5} 
\begin{array}{c|ccc|c}
    & \ket{0} & \ket{1} & \ket{2} & \sum \\
    \hline
    \ket{0} & \frac{1}{12} & \frac{1}{6} & \frac{1}{4} & \mathbf{1/2} \\
    \ket{1} & \frac{1}{4}& \frac{1}{6}& \frac{1}{12} & \mathbf{1/2} \\
    \hline
     \sum & \mathbf{1/3} & \mathbf{1/3} & \mathbf{1/3} & \mathbf{1}
\end{array}
\end{equation*}
The partitions are 
\begin{itemize}
    \item Partition $\ket{0}$ (yielding $\psi_0$): $a_{02}=\frac{1}{4}$,
    \item Partition $\ket{1}$ (yielding $\psi_1$): $a_{10}=\frac{1}{4}$,
    \item Partition $\ket{2}$ (yielding $\psi_2$): $a_{00}+a_{11}=\frac{1}{12} + \frac{1}{6} =\frac{1}{4}$,
    \item Partition $\ket{3}$ (yielding $\psi_3$): $a_{01}+a_{12}=\frac{1}{6} + \frac{1}{12} = \frac{1}{4}$.
\end{itemize}
Notice how perfectly this structure satisfies the requirements. The row sums ($1/d_1 = 1/2$) and column sums ($1/d_2 = 1/3$) are correct, and the labels confirm the normalization constraint ($1/d_3 = 1/4$), dividing the weights into four distinct sets.

Therefore our states are recovered as $\ket{\psi_i} = \sqrt{d_3} \sum_{j,k} (\Phi_i)_{jk} \ket{j}_1\ket{k}_2$ for each partition $i\in\{0,1,2,3\}$:
\begin{align*}
\ket{\psi_0}&=\sqrt{d_3}\sqrt{a_{02}}\ket{02}=2\frac{1}{\sqrt{4}}\ket{02}=\ket{02},\\
\ket{\psi_1}&=\sqrt{d_3}\sqrt{a_{10}}\ket{10}=2\frac{1}{\sqrt{4}}\ket{10}=\ket{10},\\
\ket{\psi_2}&=\sqrt{d_3}\prt{\sqrt{a_{00}}\ket{00}+\sqrt{a_{11}}\ket{11}}=2\prt{\frac{1}{\sqrt{12}}\ket{00}+\frac{1}{\sqrt{6}}\ket{11}}=\frac{1}{\sqrt{3}}\ket{00}+\sqrt{\frac{2}{3}}\ket{11},\\
\ket{\psi_3}&=\sqrt{d_3}\prt{\sqrt{a_{01}}\ket{01}+\sqrt{a_{12}}\ket{12}}=2\prt{\frac{1}{\sqrt{6}}\ket{01}+\frac{1}{\sqrt{12}}\ket{12}}=\sqrt{\frac{2}{3}}\ket{01}+\frac{1}{\sqrt{3}}\ket{12}.
\end{align*}
And our AME state is
\begin{equation}\label{eq:AME234}
\ket{\psi}_{2,3,4} = \frac{1}{2} \ket{020} + \frac{1}{2} \ket{101} + \frac{1}{2\sqrt{3}}\ket{002} + \frac{1}{\sqrt{6}}\ket{112} + \frac{1}{\sqrt{6}}\ket{013} + \frac{1}{2\sqrt{3}}\ket{123}.
\end{equation}
The complete enumerator computation is given in Table \ref{tab:AME234_enumerators}.
\begin{table}[ht]
    \centering
    \footnotesize
    \setlength{\tabcolsep}{4pt}
    \renewcommand{\arraystretch}{1.5}
    \caption{Complete enumerator profile for the mixed-dimensional AME state in $\mathbb{C}^2 \otimes \mathbb{C}^3 \otimes \mathbb{C}^4$ from Eq. \eqref{eq:AME234}. Subsystems $S \subseteq \{1, 2, 3\}$ are grouped by their corresponding dimension multiset $\mathbf{v}$, where party 1 is the qubit, party 2 is the qutrit and party 3 is the ququart.}
    \label{tab:AME234_enumerators}
    \begin{tabular}{ccccccccc}
        \toprule
        \textbf{Subsystems} $S$ & \textbf{Multiset} $\mathcal{D}(S)=\mathbf{v}$ & $A_{\mathbf{v}}$ & $B_{\mathbf{v}}$ & $S_{\mathbf{v}}$ & $A'_{\mathbf{v}}$ & $B'_{\mathbf{v}}$ & $\mathcal{A}'_S$ & $\mathcal{B}'_S$ \\
        \midrule
        $\emptyset$ & $\emptyset$ & $1$ & $1$ & $0$ & $1$ & $1$ & $1$ & $1$ \\
        $\{1\}$ & $\{2\}$ & $0$ & $0$ & $\frac{11}{6}$ & $\frac{1}{2}$ & $\frac{1}{2}$ & $\frac{1}{2}$ & $\frac{1}{2}$ \\
        $\{2\}$ & $\{3\}$ & $0$ & $0$ & $\frac{7}{6}$ & $\frac{1}{3}$ & $\frac{1}{3}$ & $\frac{1}{3}$ & $\frac{1}{3}$ \\
        $\{3\}$ & $\{4\}$ & $0$ & $0$ & $\frac{5}{6}$ & $\frac{1}{4}$ & $\frac{1}{4}$ & $\frac{1}{4}$ & $\frac{1}{4}$ \\
        $\{1, 2\}$ & $\{2, 3\}$ & $\frac{1}{2}$ & $\frac{1}{2}$ & $0$ & $\frac{1}{4}$ & $\frac{1}{4}$ & $\frac{1}{4}$ & $\frac{1}{4}$ \\
        $\{1, 3\}$ & $\{2, 4\}$ & $\frac{5}{3}$ & $\frac{5}{3}$ & $0$ & $\frac{1}{3}$ & $\frac{1}{3}$ & $\frac{1}{3}$ & $\frac{1}{3}$ \\
        $\{2, 3\}$ & $\{3, 4\}$ & $5$ & $5$ & $0$ & $\frac{1}{2}$ & $\frac{1}{2}$ & $\frac{1}{2}$ & $\frac{1}{2}$ \\
        $\{1, 2, 3\}$ & $\{2, 3, 4\}$ & $\frac{95}{6}$ & $\frac{95}{6}$ & $\frac{25}{6}$ & $1$ & $1$ & $1$ & $1$ \\
        \bottomrule
    \end{tabular}
\end{table}
\end{example}
We can also have partially symmetric AME states in $\mathbb{C}^2\otimes\mathbb{C}^2\otimes\mathbb{C}^3$ and $\mathbb{C}^2\otimes\mathbb{C}^3\otimes\mathbb{C}^3$.
\begin{example}
The combinatorial grid and the complete enumerator profile for an AME state in $\mathbb{C}^2\otimes\mathbb{C}^2\otimes\mathbb{C}^3$ are presented side-by-side in Table \ref{tab:AME223}. This results in the AME state
\begin{equation}\label{eq:AME223}
|\psi\rangle_{2,2,3} = \frac{1}{\sqrt{3}} |010\rangle + \frac{1}{\sqrt{3}} |101\rangle + \frac{1}{\sqrt{6}} |002\rangle + \frac{1}{\sqrt{6}} |112\rangle.
\end{equation}

\end{example}
\begin{example}
The combinatorial grid and the complete enumerator profile for an AME state in $\mathbb{C}^2\otimes\mathbb{C}^3\otimes\mathbb{C}^3$ are presented side-by-side in Table \ref{tab:AME233}. This results in the AME state
\begin{equation}\label{eq:AME233}
|\psi\rangle_{2,3,3}= \frac{1}{\sqrt{6}} |020\rangle + \frac{1}{\sqrt{6}} |100\rangle + \frac{1}{\sqrt{6}} |011\rangle + \frac{1}{\sqrt{6}} |121\rangle + \frac{1}{\sqrt{6}} |002\rangle + \frac{1}{\sqrt{6}} |112\rangle.
\end{equation}
\end{example}

\section{Conclusions}
In this work, we have established a comprehensive mathematical framework for the characterization and bounding of quantum error-correcting codes (QECC) and absolutely maximally entangled (AME) states in mixed-dimensional Hilbert spaces. By introducing the concept of dimension multisets, we moved beyond traditional scalar metrics of error weight, allowing for a fine-grained tracking of the physical composition of error supports across diverse quantum architectures.

Our central result, the mixed-dimensional quantum MacWilliams identity, provides the necessary algebraic link between multivariate Shor-Laflamme and unitary weight enumerators. We noticed this approach was more efficient and required less terms than previous ones. This foundation enabled the derivation of the mixed-dimensional shadow identity and the formulation of a linear program capable of systematically evaluating the feasibility of heterogeneous code parameters. 

The multiset approach proved instrumental not only in generalizing weight enumerator theory but also in establishing the quantum Hamming bound for non-degenerate mixed-dimensional codes, effectively generalizing the geometric volume-packing constraints of homogeneous systems. Furthermore, a significant physical insight emerged from our analysis of the mixed-dimensional quantum Singleton bound. We demonstrated that for pure mixed-dimensional codes, the parameter space is more severely restricted than previously assumed, resulting in a tighter bound with no direct homogeneous analogue. 

This leads to the consequence that in mixed-dimensional systems where the general Singleton bound strictly exceeds the pure-code limit, any valid quantum code saturating this bound is mathematically forced to be impure. This is a departure from the behavior of standard homogeneous QMDS codes, which are forced to be pure. Consequently, it raises an important open question for future study: what are the implications of these impurity constraints for QMDS codes in mixed-dimensional Hilbert spaces?

Finally, we applied this multiset machinery to the existence landscape of AME states. By leveraging generalized shadow inequalities, we ruled out several families of qubit-qutrit and qutrit-ququart states. To complement these non-existence results, we introduced a combinatorial grid methodology for the explicit construction of tripartite AME states, providing a bridge between quantum algebraic constraints and intuitive combinatorial design.

As quantum hardware continues to evolve toward heterogeneous architectures, integrating distinct physical substrates for logic, memory and communication, the multiset framework presented here provides the essential toolkit for designing robust, optimized error-correction protocols. Future research may extend these combinatorial grid methods to multipartite systems and explore the fault-tolerant properties of the impure codes identified by our framework.
\begin{acknowledgments}
We thank Felix Huber for providing Reference \cite{Nebe2006} and the anonymous referee for their comments. Simeon Ball acknowledges the support of the Spanish Ministry of Science, Innovation and Universities grant PID2023-147202NB-I00.
\end{acknowledgments}

\newpage
\appendix
\newpage
\appendix
\section{AME states from qutrits and ququarts}\label{sec:qutrits_ququarts}
We give an insight of the actual existence landscape of AME states formed by qutrits and ququarts. Our grid construction from Theorem \ref{thm:grid_tripartiteAME} yields AME states in $\mathbb{C}^3\otimes\mathbb{C}^3\otimes\mathbb{C}^4$ and $\mathbb{C}^3\otimes\mathbb{C}^4\otimes\mathbb{C}^4$, which are
\begin{align}
    |\psi\rangle_{3,3,4} &= \frac{1}{2\sqrt{3}} |000\rangle + \sqrt{\frac{1}{6}} |120\rangle + \frac{1}{2}|211\rangle + \frac{1}{2\sqrt{3}} |012\rangle + \sqrt{\frac{1}{6}} |102\rangle + \sqrt{\frac{1}{6}} |023\rangle + \frac{1}{2\sqrt{3}} |203\rangle,\label{eq:AME334}\\ 
    |\psi\rangle_{3,4,4} &= \frac{1}{2\sqrt{3}} |000\rangle + \sqrt{\frac{1}{6}} |120\rangle + \frac{1}{2}|211\rangle + \sqrt{\frac{1}{6}} |102\rangle + \frac{1}{2\sqrt{3}} |222\rangle + \frac{1}{2}|033\rangle.\label{eq:AME344}
\end{align}
The corresponding grids and enumerator profile can be found in Table \ref{tab:AME334} and Table \ref{tab:AME344}, respectively.

From Fig. \ref{fig:ame_heatmap_34} we can see that AME states formed by qutrits and ququarts for 8 parties can only exist in $\prt{\mathbb{C}^4}^{\otimes 8}$ or $\mathbb{C}^3\otimes \prt{\mathbb{C}^4}^{\otimes 7}$. According to \cite{TableAME}, the former has not been found yet. Also, no AME states exist for 12 parties formed by qutrits and ququarts.
\begin{figure}[ht]
    \centering
    \includegraphics[width=0.81\textwidth]{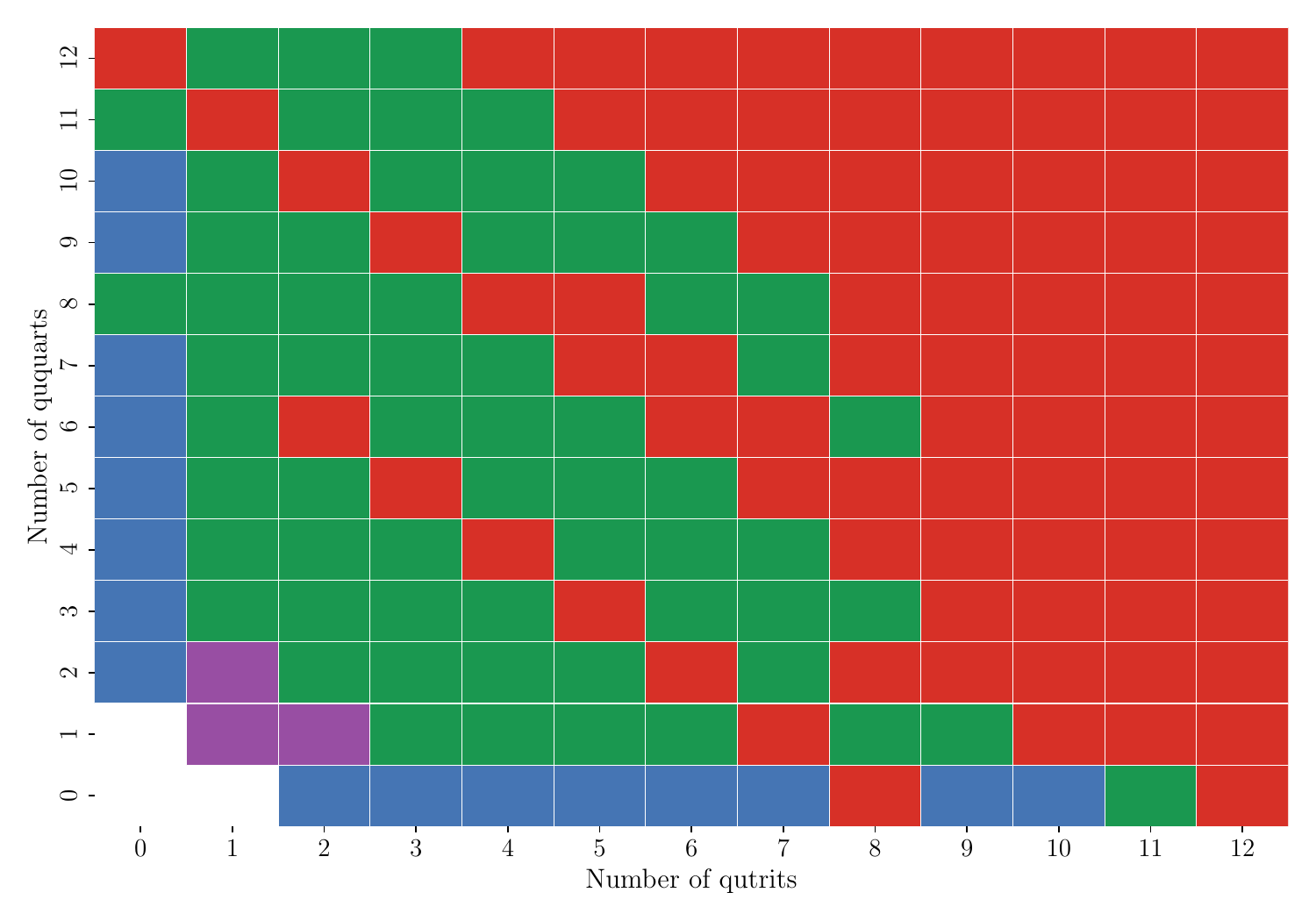}
    \caption{Existence of mixed-dimensional AME states based on the positivity of shadow multiset enumerators for qutrits and ququarts. Green squares represent potentially existing states, while red squares indicate states strictly forbidden by negative shadow enumerators. Blue squares are already existing states in homogeneous systems \cite{TableAME, huber_quantum_code_bounds}. Purple squares represent existing AME states in heterogeneous scenarios given in Eqs. \eqref{eq:AME334} and \eqref{eq:AME344}. The AME formed by a single qutrit and single ququart is trivial.}
    \label{fig:ame_heatmap_34}
\end{figure}
\newpage
\section{Tripartite AME grids and enumerator profiles}
We provide the combinatorial grids and the complete weight enumerator profiles for the tripartite AME states constructed in Section \ref{sec:tripartiteAME}.
\begin{table}[ht]
    \centering
    \caption{\textbf{(a)} Combinatorial grid for the AME state in $\mathbb{C}^2 \otimes \mathbb{C}^2 \otimes \mathbb{C}^3$ (Eq. \eqref{eq:AME223}) with the partitions $\{a_{01}\},\{a_{10}\}$ and $\{a_{00},a_{11}\}$. \textbf{(b)} Complete enumerator profile. Subsystems $S \subseteq \{1, 2, 3\}$ are grouped by their corresponding dimension multiset $\mathbf{v}$, where parties 1 and 2 are qubits, and party 3 is a qutrit.}
    \label{tab:AME223}
    \begin{minipage}[c]{0.3\textwidth}
        \centering
        \textbf{(a)}\par\vspace{1ex}
        \renewcommand{\arraystretch}{1.5} 
        $ \begin{array}{c|cc|c}
            & \ket{0} & \ket{1} &  \sum \\
            \hline
            \ket{0} & \frac{1}{6} & \frac{1}{3}  & \mathbf{1/2} \\
            \ket{1} & \frac{1}{3} & \frac{1}{6} & \mathbf{1/2} \\
            \hline
             \sum & \mathbf{1/2} & \mathbf{1/2}  & \mathbf{1}
        \end{array} $
    \end{minipage}\hfill
    \begin{minipage}[c]{0.68\textwidth}
        \centering
        \textbf{(b)}\par\vspace{1ex}
        \footnotesize
        \setlength{\tabcolsep}{3pt}
        \renewcommand{\arraystretch}{1.5}
        \begin{tabular}{ccccccccc}
            \toprule
            \textbf{Subsystems} $S$ & \textbf{Multiset} $\mathbf{v}$ & $A_{\mathbf{v}}$ & $B_{\mathbf{v}}$ & $S_{\mathbf{v}}$ & $A'_{\mathbf{v}}$ & $B'_{\mathbf{v}}$ & $\mathcal{A}'_S$ & $\mathcal{B}'_S$ \\
            \midrule
            $\emptyset$ & $\emptyset$ & $1$ & $1$ & $0$ & $1$ & $1$ & $1$ & $1$ \\
            $\{1\}$, $\{2\}$ & $\{2\}$ & $0$ & $0$ & $\frac{8}{3}$ & $1$ & $1$ & $\frac{1}{2}$ & $\frac{1}{2}$ \\
            $\{3\}$ & $\{3\}$ & $0$ & $0$ & $\frac{2}{3}$ & $\frac{1}{3}$ & $\frac{1}{3}$ & $\frac{1}{3}$ & $\frac{1}{3}$ \\
            $\{1, 2\}$ & $\{2, 2\}$ & $\frac{1}{3}$ & $\frac{1}{3}$ & $0$ & $\frac{1}{3}$ & $\frac{1}{3}$ & $\frac{1}{3}$ & $\frac{1}{3}$ \\
            $\{1, 3\}$, $\{2, 3\}$ & $\{2, 3\}$ & $4$ & $4$ & $0$ & $1$ & $1$ & $\frac{1}{2}$ & $\frac{1}{2}$ \\
            $\{1, 2, 3\}$ & $\{2, 2, 3\}$ & $\frac{20}{3}$ & $\frac{20}{3}$ & $\frac{14}{3}$ & $1$ & $1$ & $1$ & $1$ \\
            \bottomrule
        \end{tabular}
    \end{minipage}
    
\end{table}
\begin{table}[ht]
    \centering
    \caption{\textbf{(a)} Combinatorial grid for the AME state in $\mathbb{C}^2 \otimes \mathbb{C}^3 \otimes \mathbb{C}^3$ (Eq. \eqref{eq:AME233}) with the partitions $\{a_{02},a_{10}\},\{a_{01},a_{12}\},$ and $\{a_{00},a_{11}\}$. \textbf{(b)} Complete enumerator profile. Subsystems $S \subseteq \{1, 2, 3\}$ are grouped by their dimension multiset $\mathbf{v}$, where party 1 is the qubit and parties 2 and 3 are qutrits.}
    \label{tab:AME233}
    \begin{minipage}[c]{0.3\textwidth}
        \centering
        \textbf{(a)}\par\vspace{1ex}
        \renewcommand{\arraystretch}{1.5} 
        $ \begin{array}{c|ccc|c}
            & \ket{0} & \ket{1} & \ket{2} & \sum \\
            \hline
            \ket{0} & \frac{1}{6}& \frac{1}{6} & \frac{1}{6} & \mathbf{1/2} \\
            \ket{1} & \frac{1}{6}& \frac{1}{6} & \frac{1}{6} & \mathbf{1/2} \\
            \hline
             \sum & \mathbf{1/3} & \mathbf{1/3} & \mathbf{1/3} & \mathbf{1}
        \end{array} $
    \end{minipage}\hfill
    \begin{minipage}[c]{0.68\textwidth}
        \centering
        \textbf{(b)}\par\vspace{1ex}
        \footnotesize
        \setlength{\tabcolsep}{3pt} 
        \renewcommand{\arraystretch}{1.5}
        \begin{tabular}{ccccccccc}
            \toprule
            \textbf{Subsystems} $S$ & \textbf{Multiset} $\mathbf{v}$ & $A_{\mathbf{v}}$ & $B_{\mathbf{v}}$ & $S_{\mathbf{v}}$ & $A'_{\mathbf{v}}$ & $B'_{\mathbf{v}}$ & $\mathcal{A}'_S$ & $\mathcal{B}'_S$ \\
            \midrule
            $\emptyset$ & $\emptyset$ & $1$ & $1$ & $0$ & $1$ & $1$ & $1$ & $1$ \\
            $\{1\}$ & $\{2\}$ & $0$ & $0$ & $\frac{5}{3}$ & $\frac{1}{2}$ & $\frac{1}{2}$ & $\frac{1}{2}$ & $\frac{1}{2}$ \\
            $\{2\}$, $\{3\}$ & $\{3\}$ & $0$ & $0$ & $2$ & $\frac{2}{3}$ & $\frac{2}{3}$ & $\frac{1}{3}$ & $\frac{1}{3}$ \\
            $\{1, 2\}$, $\{1, 3\}$ & $\{2, 3\}$ & $2$ & $2$ & $0$ & $\frac{2}{3}$ & $\frac{2}{3}$ & $\frac{1}{3}$ & $\frac{1}{3}$ \\
            $\{2, 3\}$ & $\{3, 3\}$ & $\frac{7}{2}$ & $\frac{7}{2}$ & $0$ & $\frac{1}{2}$ & $\frac{1}{2}$ & $\frac{1}{2}$ & $\frac{1}{2}$ \\
            $\{1, 2, 3\}$ & $\{2, 3, 3\}$ & $\frac{23}{2}$ & $\frac{23}{2}$ & $\frac{13}{3}$ & $1$ & $1$ & $1$ & $1$ \\
            \bottomrule
        \end{tabular}
    \end{minipage}
    
\end{table}
\begin{table}[ht]
    \centering
    \caption{\textbf{(a)} Combinatorial grid for the AME state in $\mathbb{C}^3 \otimes \mathbb{C}^3 \otimes \mathbb{C}^4$ (Eq. \eqref{eq:AME334}) with the partitions $\{a_{00},a_{12}\}, \{a_{21}\}, \{a_{01},a_{10}\},$ and $\{a_{02},a_{20}\}$. \textbf{(b)} Complete enumerator profile. Subsystems $S \subseteq \{1, 2, 3\}$ are grouped by their dimension multiset $\mathbf{v}$, where parties 1 and 2 are qutrits and party 3 is the ququart.}
    \label{tab:AME334}
    \begin{minipage}[c]{0.3\textwidth}
        \centering
        \textbf{(a)}\par\vspace{1ex}
        \renewcommand{\arraystretch}{1.5} 
        $ \begin{array}{c|ccc|c}
            & \ket{0} & \ket{1} & \ket{2} & \sum \\
            \hline
            \ket{0} & \frac{1}{12} & \frac{1}{12} & \frac{1}{6} & \mathbf{1/3} \\
            \ket{1} & \frac{1}{6} & 0 & \frac{1}{6}& \mathbf{1/3} \\
            \ket{2} & \frac{1}{12} & \frac{1}{4} & 0 & \mathbf{1/3} \\
            \hline
             \sum & \mathbf{1/3} & \mathbf{1/3} & \mathbf{1/3} & \mathbf{1}
        \end{array} $
    \end{minipage}\hfill
    \begin{minipage}[c]{0.68\textwidth}
        \centering
        \textbf{(b)}\par\vspace{1ex}
        \footnotesize
        \setlength{\tabcolsep}{3pt} 
        \renewcommand{\arraystretch}{1.5}
        \begin{tabular}{ccccccccc}
            \toprule
            \textbf{Subsystems} $S$ & \textbf{Multiset} $\mathbf{v}$ & $A_{\mathbf{v}}$ & $B_{\mathbf{v}}$ & $S_{\mathbf{v}}$ & $A'_{\mathbf{v}}$ & $B'_{\mathbf{v}}$ & $\mathcal{A}'_S$ & $\mathcal{B}'_S$ \\
            \midrule
            $\emptyset$ & $\emptyset$ & $1$ & $1$ & $0$ & $1$ & $1$ & $1$ & $1$ \\
            $\{1\}, \{2\}$ & $\{3\}$ & $0$ & $0$ & $3$ & $\frac{2}{3}$ & $\frac{2}{3}$ & $\frac{1}{3}$ & $\frac{1}{3}$ \\
            $\{3\}$ & $\{4\}$ & $0$ & $0$ & $\frac{7}{6}$ & $\frac{1}{4}$ & $\frac{1}{4}$ & $\frac{1}{4}$ & $\frac{1}{4}$ \\
            $\{1, 2\}$ & $\{3, 3\}$ & $\frac{5}{4}$ & $\frac{5}{4}$ & $0$ & $\frac{1}{4}$ & $\frac{1}{4}$ & $\frac{1}{4}$ & $\frac{1}{4}$ \\
            $\{1, 3\}, \{2, 3\}$ & $\{3, 4\}$ & $6$ & $6$ & $0$ & $\frac{2}{3}$ & $\frac{2}{3}$ & $\frac{1}{3}$ & $\frac{1}{3}$ \\
            $\{1, 2, 3\}$ & $\{3, 3, 4\}$ & $\frac{111}{4}$ & $\frac{111}{4}$ & $\frac{23}{6}$ & $1$ & $1$ & $1$ & $1$ \\
            \bottomrule
        \end{tabular}
    \end{minipage}
    
\end{table}
\begin{table}[ht]
    \centering
    \caption{\textbf{(a)} Combinatorial grid for the AME state in $\mathbb{C}^3 \otimes \mathbb{C}^4 \otimes \mathbb{C}^4$ (Eq. \eqref{eq:AME344}) with the partitions $\{a_{00},a_{12}\}, \{a_{21}\}, \{a_{10},a_{22}\},$ and $\{a_{03}\}$. \textbf{(b)} Complete enumerator profile. Subsystems $S \subseteq \{1, 2, 3\}$ are grouped by their dimension multiset $\mathbf{v}$, where party 1 is a qutrit and parties 2 and 3 are ququarts.}
    \label{tab:AME344}
    \begin{minipage}[c]{0.35\textwidth}
        \centering
        \textbf{(a)}\par\vspace{1ex}
        \renewcommand{\arraystretch}{1.5} 
        $ \begin{array}{c|cccc|c}
            & \ket{0} & \ket{1} & \ket{2} & \ket{3} & \sum \\
            \hline
            \ket{0} & \frac{1}{12} & 0 & 0 & \frac{1}{4} & \mathbf{1/3} \\
            \ket{1} & \frac{1}{6} & 0 & \frac{1}{6} & 0 & \mathbf{1/3} \\
            \ket{2} & 0 & \frac{1}{4} & \frac{1}{12} & 0 & \mathbf{1/3} \\
            \hline
             \sum & \mathbf{1/4} & \mathbf{1/4} & \mathbf{1/4} & \mathbf{1/4} & \mathbf{1}
        \end{array} $
    \end{minipage}\hfill
    \begin{minipage}[c]{0.63\textwidth}
        \centering
        \textbf{(b)}\par\vspace{1ex}
        \footnotesize
        \setlength{\tabcolsep}{3pt} 
        \renewcommand{\arraystretch}{1.5}
        \begin{tabular}{ccccccccc}
            \toprule
            \textbf{Subsystems} $S$ & \textbf{Multiset} $\mathbf{v}$ & $A_{\mathbf{v}}$ & $B_{\mathbf{v}}$ & $S_{\mathbf{v}}$ & $A'_{\mathbf{v}}$ & $B'_{\mathbf{v}}$ & $\mathcal{A}'_S$ & $\mathcal{B}'_S$ \\
            \midrule
            $\emptyset$ & $\emptyset$ & $1$ & $1$ & $0$ & $1$ & $1$ & $1$ & $1$ \\
            $\{1\}$ & $\{3\}$ & $0$ & $0$ & $\frac{5}{3}$ & $\frac{1}{3}$ & $\frac{1}{3}$ & $\frac{1}{3}$ & $\frac{1}{3}$ \\
            $\{2\}, \{3\}$ & $\{4\}$ & $0$ & $0$ & $\frac{8}{3}$ & $\frac{1}{2}$ & $\frac{1}{2}$ & $\frac{1}{4}$ & $\frac{1}{4}$ \\
            $\{1, 2\}, \{1, 3\}$ & $\{3, 4\}$ & $4$ & $4$ & $0$ & $\frac{1}{2}$ & $\frac{1}{2}$ & $\frac{1}{4}$ & $\frac{1}{4}$ \\
            $\{2, 3\}$ & $\{4, 4\}$ & $\frac{13}{3}$ & $\frac{13}{3}$ & $0$ & $\frac{1}{3}$ & $\frac{1}{3}$ & $\frac{1}{3}$ & $\frac{1}{3}$ \\
            $\{1, 2, 3\}$ & $\{3, 4, 4\}$ & $\frac{116}{3}$ & $\frac{116}{3}$ & $\frac{11}{3}$ & $1$ & $1$ & $1$ & $1$ \\
            \bottomrule
        \end{tabular}
    \end{minipage}
\end{table}
\newpage

\end{document}